\documentclass[
notitlepage
,floatfix,
aps,
prl,
reprint,
superscriptaddress,
]{revtex4-1}
\usepackage[utf8]{inputenc}

\usepackage{graphicx} 
\usepackage{epstopdf}
\usepackage{array}
\usepackage{verbatim}
\usepackage{amsmath,amsfonts,amssymb,amscd,mathtools}
\usepackage{amsthm}
\usepackage{tabularx}
\usepackage{stmaryrd}
\usepackage{enumerate}
\usepackage[ruled]{algorithm2e}
\usepackage{wasysym}
\usepackage{braket}
\usepackage{mathrsfs}
\usepackage{microtype}
\usepackage{hyperref}
\usepackage[dvipsnames]{xcolor}
\hypersetup{
    bookmarksnumbered=true, % If Acrobat bookmarks are requested, include section numbers
    unicode=false, % non-Latin characters in Acrobat bookmarks
    pdfstartview={FitH}, % fits the width of the page to the window
    pdftitle={}, % title
    pdfauthor={}, % author
    pdfsubject={}, % subject of the document
    pdfcreator={}, % creator of the document
    pdfproducer={}, % producer of the document
    pdfkeywords={}, % list of keywords
    pdfnewwindow=true, % links in new window
    colorlinks=true, % false: boxed links; true: colored links
    linkcolor=NavyBlue, % color of internal links
    citecolor=NavyBlue, % color of links to bibliography
    filecolor=NavyBlue, % color of file links
    urlcolor=NavyBlue % color of external links
}

\usepackage{newtxtext,newtxmath}

\theoremstyle{plain}
\newtheorem{thm}{Theorem}

\newtheorem{lem}[thm]{Lemma}
\newtheorem{pro}[thm]{Proposition}

\theoremstyle{definition}
\newtheorem{defn}[thm]{Definition}
\newtheorem{remark}[thm]{Remark}

\newtheorem{conj}[thm]{Conjecture}

\newcommand{\eq}[1]{(\hyperref[eq:#1]{\ref*{eq:#1}})}

\renewcommand{\sec}[1]{\hyperref[sec:#1]{Section~\ref*{sec:#1}}}
\newcommand{\thrm}[1]{\hyperref[thrm:#1]{Theorem~\ref*{thrm:#1}}}
\newcommand{\lemm}[1]{\hyperref[lemm:#1]{Lemma~\ref*{lemm:#1}}}
\newcommand{\prop}[1]{\hyperref[prop:#1]{Proposition~\ref*{prop:#1}}}
\newcommand{\corr}[1]{\hyperref[corr:#1]{Corollary~\ref*{corr:#1}}}
\newcommand{\fig}[1]{\hyperref[fig:#1]{~\ref*{fig:#1}}}
\newcommand{\deff}[1]{\hyperref[deff:#1]{~\ref*{deff:#1}}}

\newcommand{\mE}{\mathcal{E}}

\newcommand{\mD}{\mathcal{D}}
\newcommand{\mI}{\mathcal{I}}
\newcommand{\mJ}{\mathcal{J}}
\newcommand{\mF}{\mathcal{F}}

\newcommand{\mH}{\mathcal{H}}
\newcommand{\mM}{\mathcal{M}}
\newcommand{\mO}{\mathcal{O}}

\newcommand{\mR}{\mathcal{R}}
\newcommand{\mS}{\mathcal{S}}

\newcommand{\mbI}{\mathbb{I}}

\newcommand{\mbZ}{\mathbb{Z}}
\newcommand{\mfR}{\mathfrak{R}}

\newcommand{\ketbra}[2]{|{#1}\rangle\!\langle{#2}|}

\newcommand{\ba}{\begin{eqnarray}}
\newcommand{\ea}{\end{eqnarray}}
\newcommand{\bann}{\begin{eqnarray*}}
\newcommand{\eann}{\end{eqnarray*}}
\newcommand{\bal}{\begin{equation}\begin{aligned}}
\newcommand{\eal}{\end{aligned}\end{equation}}
\newcommand{\dm}[1]{\ketbra{#1}{#1}}

\newcolumntype{L}[1]{>{\raggedright}p{#1}}
\newcolumntype{C}[1]{>{\centering}p{#1}}
\newcolumntype{R}[1]{>{\raggedleft}p{#1}}
\newcolumntype{D}{>{\centering\arraybackslash}X}

\newcommand{\sbar}{\;\rule{0pt}{9.5pt}\right|\;}

\newcommand{\lset}{\left\{\left.}
\newcommand{\rset}{\right\}}

\DeclareMathAlphabet{\matheu}{U}{eus}{m}{n}

\DeclareMathOperator{\Tr}{Tr}
\DeclareMathOperator{\id}{id}

\DeclareMathOperator{\supp}{supp}

\newcommand{\Di}{\mathit{\Delta}}

\begin{document}

%TC:ignore 
%Exclude below from the word count.

\title{Correlation in Catalysts Enables Arbitrary Manipulation of Quantum Coherence}

\author{Ryuji Takagi}
\email{ryuji.takagi@ntu.edu.sg}
\affiliation{Nanyang Quantum Hub, School of Physical and Mathematical Sciences, Nanyang Technological University, 637371, Singapore}

\author{Naoto Shiraishi}
\email{naoto.shiraishi@gakushuin.ac.jp}
\affiliation{Department of Physics, Gakushuin University, 1-5-1 Mejiro, Toshima-ku, Tokyo 171-8588, Japan}

\begin{abstract}
Quantum resource manipulation may include an ancillary state called a catalyst, which aids the transformation while restoring its original form at the end, and characterizing the enhancement enabled by catalysts is essential to reveal the ultimate manipulability of the precious resource quantity of interest. Here, we show that allowing correlation among multiple catalysts can offer arbitrary power in the manipulation of quantum coherence. We prove that \emph{any} state transformation can be accomplished with an arbitrarily small error by covariant operations with catalysts that may create a correlation within them while keeping their marginal states intact. This presents a new type of embezzlement-like phenomenon, in which the resource embezzlement is attributed to the correlation generated among multiple catalysts. We extend our analysis to general resource theories and provide conditions for feasible transformations assisted by catalysts that involve correlation, putting a severe restriction on other quantum resources for showing this anomalous enhancement, as well as characterizing achievable transformations in relation to their asymptotic state transformations. Our results provide not only a general overview of the power of correlation in catalysts but also a step toward the complete characterization of the resource transformability in quantum thermodynamics with correlated catalysts.

\end{abstract}

%TC:endignore 

\maketitle

\textit{\textbf{Introduction.}}
---
Quantum superposition, also known as \emph{quantum coherence}, is one of the most striking quantum features and also a useful operational resource in quantum metrology~\cite{Giovannetti2006quantum}, quantum clock~\cite{Janzing2003quasi}, and work extraction~\cite{lostaglio_description_2015}.
In quantum thermodynamics, the presence of coherence is considered as the main source of difference between semiclassical and quantum setups~\cite{lostaglio2019introductory}.
Under the presence of a conserved quantity such as Hamiltonian, one is restricted to the operations that cannot create coherence. 
These operations, known as \emph{covariant} operations, are subject to many restrictions originating from the superselection rule~\cite{Bartlett2007reference,marvian_coherence_2020,Ozawa2002conservative,Tajima2018uncertainty,Takagi2020universal,Chiribella2021fundamental,tajima2021symmetry}, while preshared coherent states can lift their operational capability~\cite{Marvian2008building,Aberg2014catalytic,Tajima2020coherence}.
This motivates us to  obtain a precise understanding of how one could quantify and efficiently manipulate coherence, for which a resource-theoretic approach has been proven useful~\cite{Gour2008resource,marvian_extending_2014,marvian2012symmetry}.

Characterizing the possible state transformations under given accessible operations is a central problem in any operational setting with physical restrictions.
To understand the fundamental resource transformability, one needs to consider an ancillary system serving as a \emph{catalyst}, which keeps its form at the end of the transformation.  
Several possible scenarios for catalytic transformations have been proposed.
The first scenario considers an \emph{uncorrelated catalyst} $\tau$ that enables the transformation from $\rho\otimes\tau$ to $\rho'\otimes\tau$~\cite{Jonathan1999entanglement,Duan2005multiple,Turgut2007catalytic,aubrun_catalytic_2008,Brandao2015second,GOUR2015resource,Bu2016catalytic,Ng2015limits,Wilming2017thirdlaw,LipkaBartosik2021all}.
Although uncorrelated catalysts can enhance state transformation in some settings such as entanglement theory~\cite{Jonathan1999entanglement,Duan2005multiple} and quantum thermodynamics~\cite{Brandao2015second,GOUR2015resource}, any pure uncorrelated catalyst fails to change the power of coherence transformation by covariant operations~\cite{Marvian2013theory,Ding2021amplifying}. 
The second scenario extends the uncorrelated catalysts by allowing correlation between the system and the catalytic system at the end of the protocol, where we consider a transformation from $\rho\otimes\tau$ to $\tilde\rho_{SC}$ such that $\Tr_C\tilde\rho_{SC}=\rho'$ and $\Tr_S\tilde\rho_{SC}=\tau$, in which we call $\tau$ a \emph{correlated catalyst}~\cite{Muller2016generalization,Wilming2017axiomatic,Muller2018correlating,Boes2019vonNeumann,wilming2020entropy,Boes2020passingfluctuation,Shiraishi2021quantum,kondra2021catalytic,lipkabartosik2021catalytic,Boes2018catalytic}.
The power of correlated catalysts in covariant operations was discussed in terms of coherence broadcasting, where it was shown that correlated catalysts do not allow covariant operations to create finite coherence from zero coherence~\cite{Lostaglio2019coherence,Marvian2019nobroadcasting}.

These observations on the limitations of catalysts in coherence transformation motivate us to investigate other forms of catalysts that could enhance covariant operations.
An interesting setting was offered in quantum thermodynamics.  
Lostaglio et al.~\cite{Lostaglio2015stochastic} considered transformations with multiple catalysts where correlation can be present \emph{among} the catalysts at the end of the transformation, i.e., from $\rho\otimes\tau_{C^{(0)}}\dots\otimes\tau_{C^{(K-1)}}$ to $\rho'\otimes\tau_{C^{(0)}\dots C^{(K-1)}}$ while the marginal state of $\tau_{C^{(0)}\dots C^{(K-1)}}$ on each catalytic system $C^{(j)}$ remains as the original catalyst $\tau_{C^{(j)}}$.
They showed that quasiclassical transformations by thermal operations~\cite{horodecki_fundamental_2013} in this form are characterized solely by the free energy, surpassing the enhancement provided by uncorrelated catalysts~\cite{Brandao2015second}. 
Although the perfect reusability is generally lost due to the correlation generated among the final state of the catalysts, we follow the terminology in Ref.~\cite{Wilming2017axiomatic} and call such a finite set of states $\otimes_{i=0}^{K-1}\tau_{C^{(j)}}$ \emph{marginal catalysts}.
Characterizing the capability of covariant operations with marginal catalysts will provide insights into an ultimate coherence manipulability, as well as differences in operational capability of covariant operations and thermal operations, the latter of which is a subclass of the former.
Although significant progress has been made for qubit coherence transformation~\cite{Ding2021amplifying}, the potential of marginal catalysts in general coherence transformation has still been left unclear.

Here, we show that correlation among catalysts can completely remove the aforementioned limitations and even provide unlimited power to coherence manipulation. 
We prove that covariant operations assisted by marginal catalysts enable \emph{any} state transformations with arbitrary precision, making a high contrast to coherence transformation with the other catalytic settings. 
Furthermore, we discuss the underlying mechanism of this phenomenon from the viewpoint of general resource theories of quantum states~\cite{Chitambar2019quantum,Horodecki2013quantumness,Brandao2015reversible,Liu2017resource,Anshu2018quantifying,Regula2017convex,Takagi2019operational,Takagi2019general,Uola2019quantifying,Liu2019oneshot,Regula2020benchmarking,Fang2020nogo,Takagi2020universal,Kuroiwa2020generalquantum,boes2020variance,Zhou2020general}.
We show that an arbitrary state transformation is forbidden in a wide class of resource theories, establishing the peculiarity of quantum coherence among other quantum resources.
We also relate single-shot catalytic transformations to the asymptotic transformation in general resource theories and exactly characterize feasible state transformations for several important settings such as quantum thermodynamics, entanglement, and speakable coherence~\cite{Baumgratz2014quantifying,Marvian2016speakable} with the resource measures based on the relative entropy~\cite{Vedral2002role}, offering them with an operational meaning in terms of extended classes of single-shot catalytic transformations.

%%%%%%%%%%%%%%%%%%%%%%%%%%%%%%%%%%%%%%%%%%%%%%%%%%%%%%%%%%%%%%%%%%%%%%%%%%%%%%%%%%%%%%%%%%%%%%%%%%%%%%%%%%%%%%%%%%%%%%%%%%%%%%%%%%%%%%%%%%%%%%%%%%

\textit{\textbf{Arbitrary state transformation.}}
---
For an arbitrary system $X$ with dimension $d_X$, let $\mD(X)$ be the set of quantum states defined in $X$ and $H_X=\sum_{i=0}^{d_X-1}E_{X,i} \dm{i}_X$ be its Hamiltonian where $\ket{i}_X$ is an energy eigenstate. 
When multiple systems $X_0,X_1,\dots,X_{N-1}$ are involved, we consider the total Hamiltonian over the systems in the additive form as $H_{X_0\dots X_{N-1}}=\sum_{i=0}^{N-1}H_{X_i}\otimes\mbI_{\bar{i}}$ where $\bar{i}$ refers to the systems other than the $i$\,th system.
Coherence between eigenstates with distinct energies can be quantitatively analyzed in the resource theory of asymmetry with U(1) group~\footnote{Although the resource theory of asymmetry can be defined for a general symmetry group~\cite{Gour2008resource}, in this manuscript we focus on the case of U(1) symmetry, which concerns the phase covariance.}.
Resource theories are frameworks accounting for the quantification and manipulation of precious quantities with respect to freely accessible quantum states and dynamics under given physical settings~\cite{Chitambar2019quantum}. 
The resource theory of asymmetry considers states without coherence, i.e., invariant under time translation, as free states and \emph{covariant} channels as free operations.
We call a channel $\mE:\mD(A)\to \mD(B)$ covariant if its action is invariant under time translation, i.e., $e^{-iH_B t}\mE(\rho)e^{iH_B t}=\mE\left(e^{-iH_At}\rho e^{iH_A t}\right)$ for any $\rho$ and $t$. 
Importantly, covariant operations cannot create coherence from incoherent states, making coherence a precious quantum resource under the situation where only covariant operations are accessible. 
Such a situation arises when energy-conserving dynamics are concerned. 
It is known that a map $\mE$ is covariant if and only if it can be implemented by an energy-conserving unitary $U_{SE}$ satisfying $[U_{SE},H_{SE}]=0$ as $\mE(\cdot)=\Tr_E[U_{SE}(\cdot\otimes\sigma) U_{SE}^\dagger]$ where $\sigma$ is an ancillary incoherent state~\cite{Keyl1999optimal,marvian2012symmetry}. 
If the ancillary state is restricted to the Gibbs state in $E$, the channels in this form coincide with the thermal operations~\cite{horodecki_fundamental_2013}. 
Therefore, covariant operations can be seen as an operation that focuses on the coherence part of the resource in quantum thermodynamics, and clarifying the difference in operational power between covariant operations and thermal operations under the same catalytic setting will help pinpoint the roles played by classical athermality and quantum coherence~\cite{Lostaglio2015quantum}.

We first formally define the marginal-catalytic covariant transformation as follows. (See also Fig.~\ref{fig:setting}.)

\begin{defn}\label{def:marginal}
$\rho\in\mD(S)$ is transformable to $\rho'\in\mD(S')$ by a marginal-catalytic covariant transformation if there exists a constant $K$ and a state $\otimes_{j=0}^{K-1}\tau_{C^{(j)}}$ in a finite-dimensional system $\otimes_{j=0}^{K-1}C^{(j)}$ and a covariant operation $\mE:\mD(SC^{(0)}\dots C^{(K-1)})\rightarrow \mD(S'C^{(0)}\dots C^{(K-1)})$ such that 
\bal
 \mE(\rho\otimes\tau_{C^{(0)}}\dots\otimes\tau_{C^{(K-1)}})&= \rho'\otimes \tau_{C^{(0)}\dots C^{(K-1)}}\\ \Tr_{\overline{C^{(j)}}}\tau_{C^{(0)}\dots C^{(K-1)}}&= \tau_{C^{(j)}},\ \forall j,
 \label{eq:marignal catalytic def}
\eal
where $\Tr_{\overline{X}}$ denotes the partial trace over the systems other than $X$.  
\end{defn}

\begin{figure}
    \centering
    \includegraphics[width=0.5\textwidth]{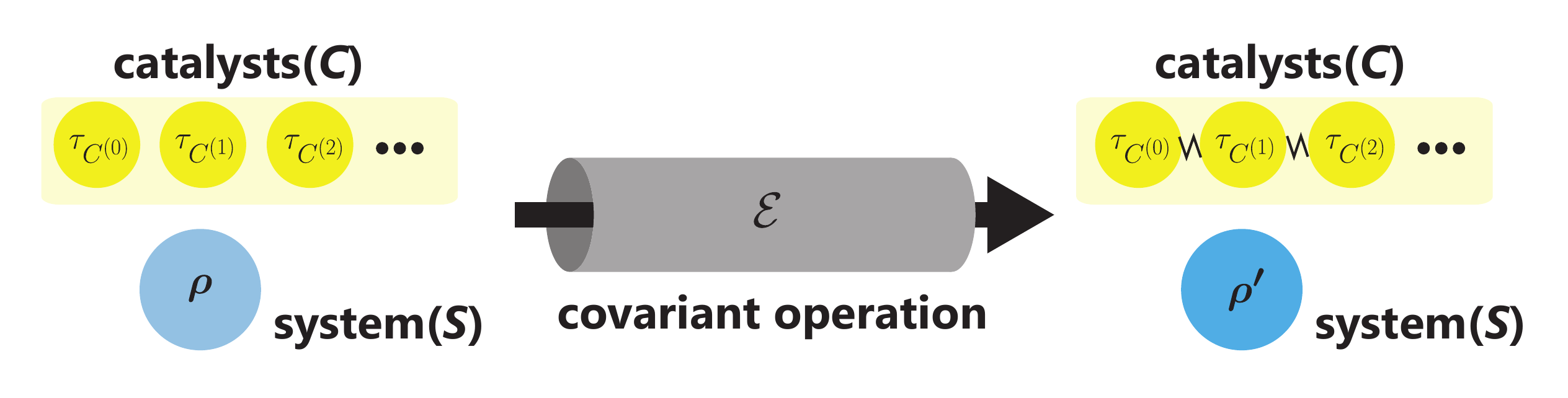}
    \caption{Schematic of marginal-catalytic covariant transformations.
    Each catalyst should get back to the original state, while it can correlate with the system and other catalysts.}
    \label{fig:setting}
\end{figure}

Marginal catalysts can be seen as an extension of a catalytic transformation in the sense that the final state keeps some properties of the initial state in an exact form.
Our main focus here is not to keep the repeatable property of catalytic transformations but to investigate how the state transformability could be enhanced by the nontrivial change of the setting regarding the correlation.
Nevertheless, we can also motivate this specific setting operationally; although the final catalyst as a whole is not reusable in the next round, \emph{some parts of it} can be reused multiple times without the degradation of performance in the desired state transformation. We discuss this \emph{partial reusability} of marginal catalysts in the Supplemental Material~\footnote{See the Supplemental Material for the full proofs of the main results and other related discussions, which includes Ref.~\cite{olver2010nist}}. 

\nocite{olver2010nist}

Our main result shows that, despite the apparent limitations in catalytic coherence transformation with uncorrelated and correlated catalysts~\cite{Marvian2013theory,Ding2021amplifying,Marvian2019nobroadcasting,Lostaglio2019coherence}, marginal catalysts can provide extraordinary power---in fact, any state transformation can be accomplished by a marginal-catalytic covariant transformation with arbitrary accuracy.

\begin{thm} \label{thm:transform catalyst}
For any $\rho\in\mD(S)$, $\rho'\in\mD(S')$ and $\epsilon>0$, $\rho$ can be transformed to a state $\rho'_\epsilon\in\mD(S')$ such that $\frac{1}{2}\|\rho'-\rho'_\epsilon\|_1\leq \epsilon$ by a marginal-catalytic covariant transformation.
\end{thm}

We sketch our proof in a later section, while deferring the detailed proof to the Supplemental Material. 

Theorem~\ref{thm:transform catalyst} implies that marginal catalysts can trivialize coherence transformations, fully generalizing the result in Ref.~\cite{Ding2021amplifying} established for the qubit state transformations to those involving arbitrary Hilbert spaces of finite dimensions.
Notably, the result can be extended to the implementation of an arbitrary \emph{quantum channel} (see the Supplemental Material). 

Our result makes a high contrast to quantum thermodynamics, in which state transformations by marginal-catalytic thermal operations respect the ordering of the free energy~\cite{Lostaglio2015stochastic}. 
A related phenomenon is known as embezzlement~\cite{vanDam2003universal}, where a negligibly small error in a catalyst enables an arbitrary transformation. 
We stress that marginal-catalytic transformations recover the marginal states \emph{exactly} and are fundamentally different from the mechanism of the well-known embezzlement. 
In fact, as discussed below, the trivialization of state transformations by marginal catalysts is an unusual phenomenon, which shows a clear contrast to the embezzlement seen in a broad class of resource theories from entanglement~\cite{vanDam2003universal} to quantum thermodynamics~\cite{Brandao2015second}.

%%%%%%%%%%%%%%%%%%%%%%%%%%%%%%%%%%%%%%%%%%%%%%%%%%%%%%%%%%%%%%%%%%%%%%%%%%%%%%%%%%%%%%%%%%%%%%%%%%%%%%%%%%%%%%%%%%%%%%%%%%%%%%%%%%%%%%%%%%%%%%%%%%%%%%%%%%%%%%%%%%%%

\textit{\textbf{Comparison to other quantum resource theories.}}
---
It may appear odd that one can create unbounded coherence in the main system while keeping the reduced states of the catalysts intact. 
To get insights into this phenomenon, let us consider whether marginal catalysts could provide similar enhancement in other quantum resource theories.
Each resource theory is equipped with a set $\mF$ of free states and a set $\mO_\mF$ of free operations~\cite{Chitambar2019quantum}. 
For given these sets, one can define a resource measure $\mfR$, which evaluates zero for any free state, i.e., $\mfR(\sigma)=0$ for any $\sigma\in\mF$, and does not increase under free operations, i.e., $\mfR(\mE(\rho))\leq \mfR(\rho)$ for any $\rho$ and for any $\mE\in\mO_\mF$.
We particularly call it superadditive if $\mfR(\rho_{12})\geq \mfR(\Tr_2[\rho_{12}])+\mfR(\Tr_1[\rho_{12}])$ for any state $\rho_{12}\in\mD(S_1\otimes S_2)$ and tensor-product additive if $\mfR(\rho_1\otimes\rho_2)=\mfR(\rho_1)+\mfR(\rho_2)$ for any $\rho_1$ and $\rho_2$. 

The setup of catalytic transformations can be extended to any resource theory. 
We say that $\rho$ is transformable to $\rho'$ by a  \emph{correlated-catalytic free transformation} if there exists a finite-dimensional catalyst $\tau$ such that $\rho\otimes\tau$ can be transformed to $\tilde\rho_{SC}$ with $\Tr_C\tilde\rho_{SC}=\rho'$,  $\Tr_S\tilde\rho_{SC}=\tau$ by a free operation.
Then, we can show that any resource measure satisfying the above two properties remains a valid resource measure under the two catalytic transformations involving correlation. (See the Supplemental Material for a proof.)

\begin{pro}\label{pro:general monotone marginal catalytic}
For any given $\mF$ and $\mO_\mF$, suppose that a resource measure $\mfR$ satisfies the superaddtivity and the tensor-product additivity.
Then, if $\rho$ is transformable to $\rho'$ by a marginal-catalytic or correlated-catalytic free transformation, $\mfR(\rho)\geq \mfR(\rho')$ holds.
\end{pro}

We remark that a related observation was made in the context of quantum thermodynamics~\cite{Wilming2017axiomatic}.
This puts a severe constraint on the possibility of arbitrary state transformation. 
If there exists even a \emph{single} resource measure satisfying the superadditivity and the tensor-product additivity, then marginal catalysts do not enable an arbitrary state transformation as long as the resource measure is faithful, i.e., any non-free state takes a non-zero value. (See also Refs.~\cite{takagi_skew_2019,Marvian2019nobroadcasting} and discussion below.)
In fact, one can find such measures in many resource theories, including quantum thermodynamics~\cite{Wilming2017axiomatic},  entanglement~\cite{Christandl2004squashed,Alicki2004continuity}, and speakable coherence (superposition between given orthogonal states)~\cite{Baumgratz2014quantifying, xi_quantum_2015,Marvian2016speakable}, prohibiting the anomalous resource transformation with marginal or correlated catalysts.

On the other hand, Theorem~\ref{thm:transform catalyst} and Proposition~\ref{pro:general monotone marginal catalytic} imply that there \emph{never} exists a coherence measure that is superadditive, tensor-product additive, and faithful.
Our results parallel previous observations; recent analytic proofs of the violation of superadditivity of coherence measures employ covariant operations that can amplify the sum of local coherence indefinitely~\cite{takagi_skew_2019,Marvian2019nobroadcasting}. 
Examples of tensor-product additive and faithful coherence measures include the Wigner-Yanase skew information~\cite{Wigner1963information,marvian_extending_2014} and other metric-adjusted skew informations~\cite{Hansen2008metric,Zhang2017detecting}, which indeed violate the superadditivity~\cite{hansen_wigner-yanase_2007,seiringer_failure_2007,takagi_skew_2019}. 
More generally, it was shown that any faithful measure of asymmetry cannot be superadditive~\cite{Marvian2019nobroadcasting}.  
These results together with Theorem~\ref{thm:transform catalyst} and Proposition~\ref{pro:general monotone marginal catalytic} indicate an intimate connection between the anomalous coherence amplification and the violation of the superadditivity of coherence measures.

Besides the necessary conditions established in Proposition~\ref{pro:general monotone marginal catalytic}, we can also formulate sufficient conditions using a general method of converting asymptotic transformations to one-shot correlated-catalytic transformations~\cite{Shiraishi2021quantum}. (See the Supplemental Material for a proof.)

\begin{pro}\label{pro:sufficient}
For any given $\mF$ and $\mO_\mF$, suppose that $\mO_\mF$ includes the relabeling of the classical register and free operations conditioned on the classical register.
Then, if $\rho$ is asymptotically transformable to $\rho'$, there exists a free transformation from $\rho$ to $\rho'$ with a correlated catalyst as well as marginal catalysts with an arbitrarily small error. 
\end{pro}

Proposition~\ref{pro:sufficient} shows that sufficient conditions for asymptotic transformations are directly carried over to single-shot catalytic transformations.
This particularly implies that, in a general class of convex resource theories, the regularized relative entropy measure provides a sufficient condition for these single-shot catalytic transformations under asymptotically resource non-generating operations, given the generalized quantum Stein's lemma holds~\cite{Brandao2015reversible,Brandao2010generalization} (see also Ref.~\cite{Berta2022on_a_gap} for the recent argument about the incompleteness in the proof of the generalized quantum Stein's lemma).

Combining Propositions~\ref{pro:general monotone marginal catalytic} and \ref{pro:sufficient}, we arrive at the complete characterizations of marginal- and correlated-catalytic free transformations for various settings in which the resource measures governing asymptotic transformations satisfy the tensor-product additivity and superadditivity.
These include several well-known relative entropy based measures, such as the free energy in quantum thermodynamics with Gibbs-preserving operations~\cite{matsumoto2010reverse,Shiraishi2021quantum,Regula2020benchmarking,Zhou2020general}, the entanglement entropy with LOCC pure state transformations~\cite{Bennett1996concentrating,lipkabartosik2021catalytic,kondra2021catalytic}, and the relative entropy of speakable coherence with several free operations~\cite{Winter2016operational,Zhao2018oneshot,Regula2018oneshot,Chitambar2018dephasing}. 
Notably, Proposition~\ref{pro:general monotone marginal catalytic} and \ref{pro:sufficient} imply the equivalence in the power of correlated and marginal catalysts for these scenarios.  

\textit{\textbf{Correlated-catalytic covariant transformations.}}
---
Although Theorem~\ref{thm:transform catalyst} reveals the exceptional power of marginal catalysts, the power of correlated catalysts in coherence transformation still remains elusive. 
In particular, when the initial state has non-zero coherence, neither the coherence no-broadcasting theorem~\cite{Marvian2019nobroadcasting,Lostaglio2019coherence} nor Proposition~\ref{pro:general monotone marginal catalytic} prohibits preparing an arbitrary state. 
We conjecture that a broad class of transformations is possible with correlated catalysts under the presence of initial coherence.

\begin{conj}[Informal] \label{conj:transform catalyst}
Let $\mathcal{C}(\rho)$ be the set of energy differences for which $\rho$ possesses non-zero coherence. Then, for any states $\rho$ and $\rho'$, $\rho$ can be transformed to $\rho'$ with an arbitrarily small error by a correlated-catalytic covariant transformation if and only if every energy difference in $\mathcal{C}(\rho')$ can be written as a sum of integer multiples of energy differences in $\mathcal{C}(\rho)$.
\end{conj}

The idea behind this is that correlated-catalytic covariant transformations should be able to amplify and manipulate the coherence of the initial state to realize any degree of coherence for the energy differences that are combinations of the initial ones with non-zero coherence.   
Thus, if these energy differences cover those of the target state with non-zero coherence, $\rho$ should be transformable to $\rho'$ under a correlated-catalytic covariant transformation.

In the Supplemental Material, we present a precise statement of the conjecture and support it by proving the state transformability under a slightly larger class of catalytic covariant operations, together with several other observations.

%%%%%%%%%%%%%%%%%%%%%%%%%%%%%%%%%%%%%%%%%%%%%%%%%%%%%%%%%%%%%%%%%%%%%%%%%%%%%%%%%%%%%%%%%%%%%%%%%%%%%%%%%%%%%%%%%%%%%%%%%%%%%%%%%%%%%%%%%%%%%%%%%%%%%%%%%%%%%%%%%%%%%%%%%%%%%%%%%%

\textit{\textbf{Proof sketch for Theorem~\ref{thm:transform catalyst}.}}
---
To prove our claim, it suffices to provide a protocol that prepares a final state from scratch with an arbitrarily small error.  
Our protocol makes use of the procedure introduced in Ref.~\cite{Ding2021amplifying} as a subroutine, which amplifies coherence in two-level systems using a correlated catalyst by a small amount (Fig.~\ref{fig:protocol1}.(a)). 
This can particularly bring a coherent state $\Sigma(\eta)\coloneqq(\mbI+\eta X)/2$ with $\eta>0$, $X\coloneqq \ketbra{0}{1}+\ketbra{1}{0}$ to another coherent state $\Sigma(\eta')$ with $\eta'>\eta$ using a catalyst $\Gamma(\eta)\coloneqq\frac{1}{2}\left(\mbI+\frac{\sqrt{3}\eta}{2}X+\frac{4-\eta^2}{6}Z\right)$, and sequential application of this protocol allows us to realize any coherent state on the $X$ axis of the Bloch sphere (excluding the pure state $\eta=1$) with marginal-catalytic covariant operations.
Although the authors of Ref.~\cite{Ding2021amplifying} claim that the whole sequence of amplification is a correlated-catalytic covariant transformation, their argument is, unfortunately, insufficient --- in fact, the total amplification process is marginal catalytic.
We extend detailed discussions about the two-level coherence amplification subroutine in the Supplemental Material.

Our protocol consists of three main steps (Figs.~\ref{fig:protocol1} and \ref{fig:protocol2}). 
The first step creates small coherence in an ancillary system, the second step amplifies this coherence and constructs coherent resource states, and the third step uses them to prepare the target state with a covariant operation.

\begin{figure}
    \centering
    \includegraphics[width=0.5\textwidth]{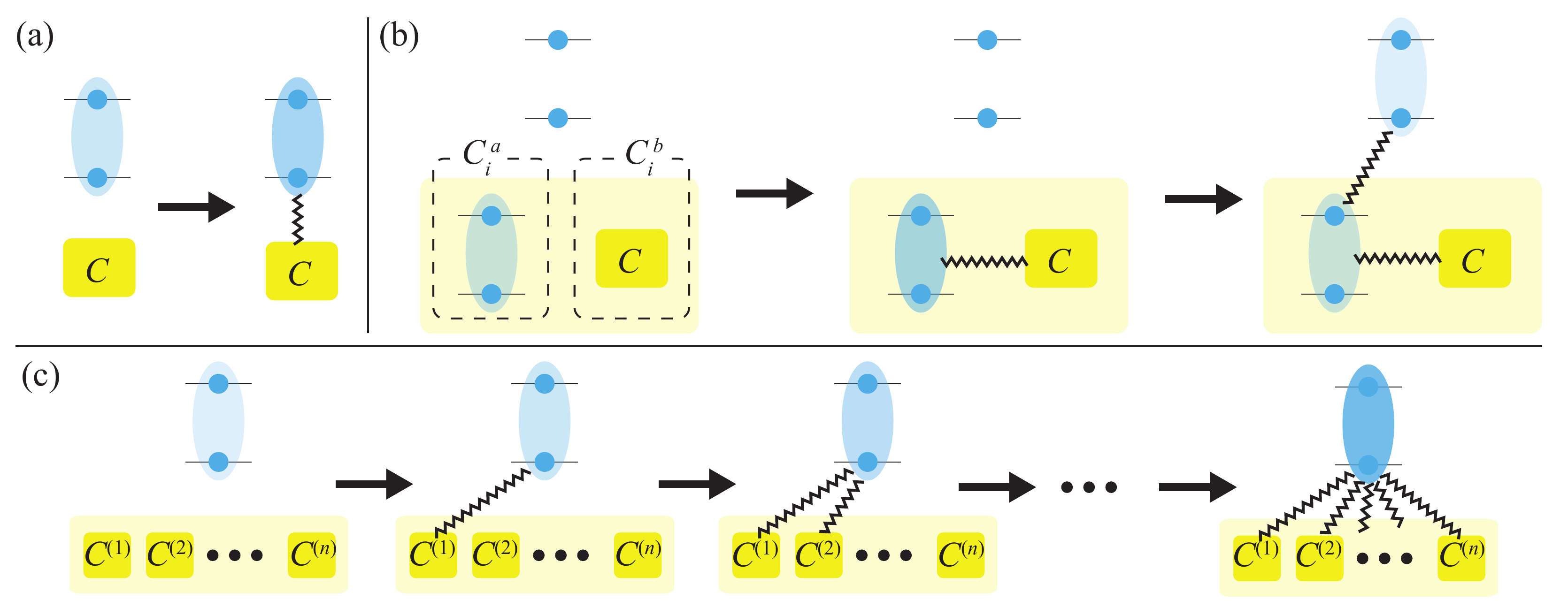}
    \caption{Schematics for Steps~1 and 2. (a) One cycle of the two-level coherence amplification subroutine. (b) We run the two-level amplification protocol to increase coherence in $C^{a}$ together with $C^{b}$. We transfer this increased coherence to $R$ and restore the state in $C^{a}$ to the original form. (c) We amplify the coherence generated in $R$ with many rounds of the two-level amplification protocol.}
    \label{fig:protocol1}
\end{figure}

\emph{Step~1: Creating small coherence. (Fig.~\ref{fig:protocol1}\,(b))}\quad  
We introduce ancillary system $R$ consisting of two-level subsystems $\{R_{i}\}_{i=1}^{d_{S'}-1}$ whose Hamiltonians reflect the spectrum of $H_{S'}$ as $H_{R_{i}}= (E_{S',i}-E_{S',j^\star})\dm{1}_{R_{i}}$ where $j^\star\in\{0,\dots,d_{S'}-1\}$ is an arbitrarily chosen integer independent of $i$.
We aim to prepare a coherent state $\Sigma(\eta)$ with $\eta>0$ for each $i$.
To this end, we introduce catalytic subsystems $C^{a}_{i}$ and $C^{b}_i$, both of which have the Hamiltonian $H_{R_{i}}$.
We prepare catalysts $\tau_i^{a}\coloneqq\Sigma(\eta)$ in $C_{i}^{a}$ and $\tau_i^{b}\coloneqq\Gamma(\eta)$ in $C^{b}_i$ for some $\eta$ with $0<\eta<1$. 
We apply a single round of the two-level coherence amplification over $\tau_i^{a}\otimes\tau_i^{b}$, which increases coherence in $C_i^{a}$ by a small amount while keeping the reduced state on $C_i^{b}$ unchanged.
We transfer the increased amount of coherence from $C_i^a$ to $R_i$ by applying a covariant unitary over $R_{i}C_i^{a}$, creating a non-zero coherence in $R_{i}$ while bringing the reduced state on $C_i^{a}$ back to $\tau_i^a$.

\emph{Step~2: Amplifying coherence.(Fig.~\ref{fig:protocol1}\,(c))}\quad   
We amplify this non-zero coherence generated in $R_i$ by the two-level coherence amplification using another set of catalysts prepared in $\otimes_{j=0}^{K-1}C_{i}^{(j)}$ with a large enough integer $K$. 
This prepares a state close to $\ket{+}:=(\ket{0}+\ket{1})/\sqrt{2}$ in $R_{i}$.

\emph{Step~3: Prepare the target state. (Fig.~\ref{fig:protocol2})}\quad
We repeat Steps~1 and 2 for $L(\gg 1)$ times to prepare a state close to $\ket{+}^{\otimes L}_{R_i}$ for each $i$, which is a superposition of energy eigenstates with weights according to the binomial distribution. 
By employing these states as ancillary coherent resource states, we can implement any unitary on $S'$ with arbitrary accuracy by a covariant operation~\cite{Aharonov1967charge,Ozawa2002conservative,Kitaev2004superselection,Bartlett2007reference,Marvian2008building,Tajima2020coherence, Aberg2014catalytic, korzekwa2016extraction}.
Since any pure state on $S'$ can be prepared by applying an appropriate unitary to an incoherent state $\ket{j^\star}_{S'}$, and any mixed state is realized by a probabilistic mixture of pure states, we can obtain the total state whose reduced state on $S'$ is $\rho'_\epsilon$.
Finally, the correlation between $S'$ and the catalytic system can be removed by using the technique employed in Ref.~\cite{Muller2016generalization}, where we start with the final state in another catalytic system and swap it with the marginal state created in $S'$.

\begin{figure}
    \centering
    \includegraphics[width=0.5\textwidth]{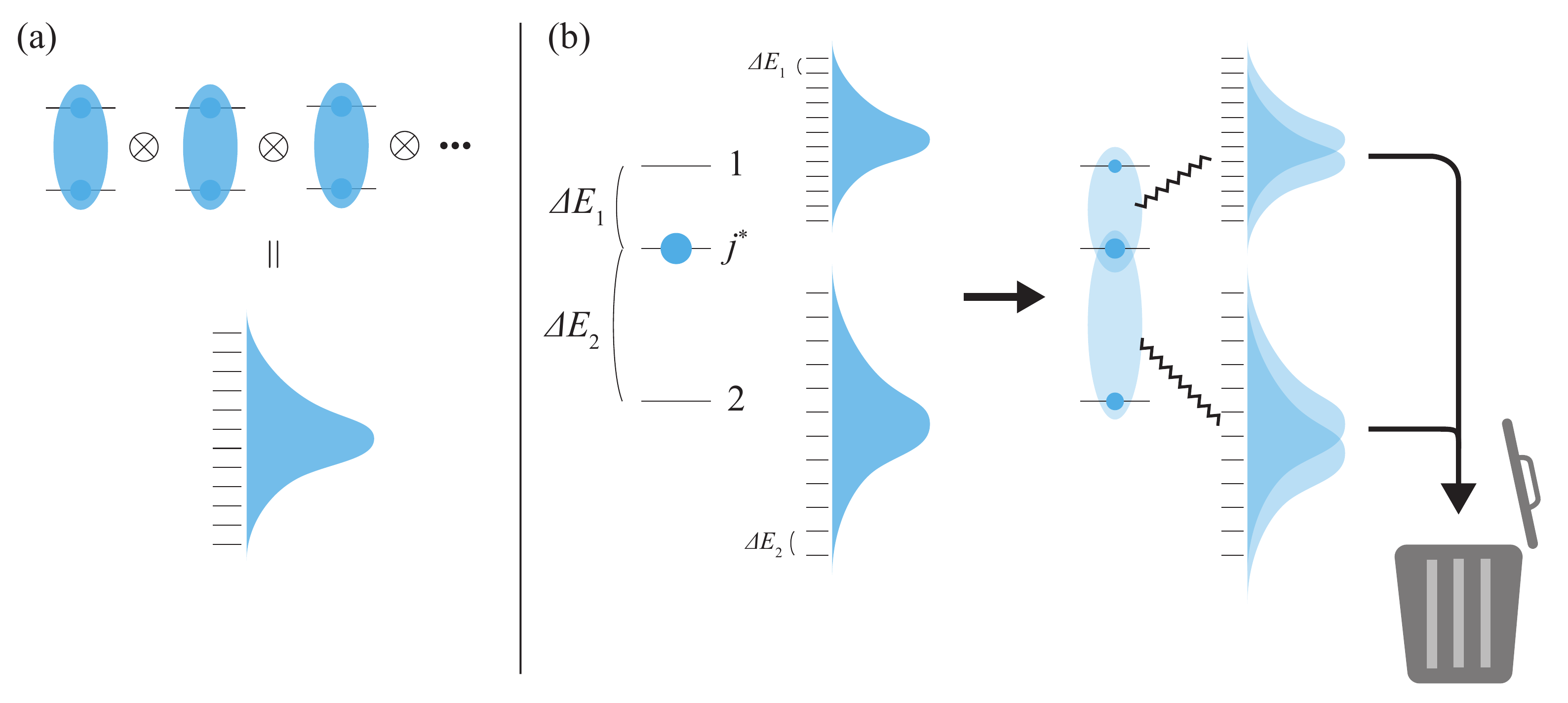}
    \caption{Schematics for Step~3. (a) Multiple copies of the resource state in $R$ created in Step~1 constitute a state with the binomially distributed energy statistics. (b) We use these coherent states as ancillary coherent resource states to assist the energy transition required for the desired unitary $V$. The error on the realized unitary from the desired one, $V$, is quantified by the overlap between the original resource states and the final resource states subject to an energy shift due to the backreaction, which can be made arbitrarily small by creating a sufficiently large resource state. These resource states are discarded at the end of the protocol.}
    \label{fig:protocol2}
\end{figure}

The accuracy of the whole protocol is determined by the errors in preparing highly coherent resource states in Step~2 and in approximating unitary in Step~3, both of which can be made arbitrarily small using finite-size catalysts, ensuring any target error $\epsilon>0$.
Also, since the required catalysts except for the final step in removing the correlation do not depend on the target state, they construct a universal family of catalysts applicable to any state transformation if an arbitrarily small correlation is allowed between the main system and catalytic systems.

%%%%%%%%%%%%%%%%%%%%%%%%%%%%%%%%%%%%%%%%%%%%%%%%%%%%%%%%%%%%%%%%%%%%%%%%%%%%%%%%%%%%%%%%%%%%%%%%%%%%%%%%%%%%%%%%%%%%%%%%%%%%%%%%%%%%%%%%%%%%%%%%%%%%%%%%%%%%%%%%%%%%%%%%%%%%%%%%%%

\textit{\textbf{Conclusions.}}
---
We studied catalytic state transformation with marginal catalysts, which allow correlation among multiple catalyst states at the end of the transformation.
We showed that marginal catalysts provide exceptional power to coherence transformation, enabling any state transformation with arbitrarily small error.
To elucidate the peculiarity of how quantum coherence behaves in catalytic transformations, we compared it to other types of quantum resources by formulating conditions for catalytic state transformations from the perspective of resource quantifiers. 
We showed that such an anomalous state transformation is impossible in resources such as thermal nonequilibrium, entanglement, and speakable coherence, for which we exactly characterized state transformability with correlated and marginal catalysts in terms of relative entropy resource measures.

An intriguing future direction is to prove or disprove the conjecture on the power of correlated catalysts in coherence transformation, which will provide insights into another interesting problem in quantum thermodynamics, that is, whether the free energy solely determines the state transformability by thermal operations with correlated catalysts if an initial state has finite coherence. 
Answering this question will pave the way toward a fully general operational characterization of single-shot quantum thermodynamics.

%TC:ignore 
%Exclude below from the word count.

\begin{acknowledgments}

\textit{Acknowledgments.}---
We thank Nelly Ng for discussions, Xueyuan Hu for comments on a preliminary version of the manuscript, and Hiroyasu Tajima for comments on the partial reusability of marginal catalysts. We also thank anonymous referees for their useful comments and suggestions.
R.T.\ acknowledges the support of National Research Foundation (NRF) Singapore, under its NRFF Fellow programme (Award No. NRF-NRFF2016-02), the Singapore Ministry of Education Tier 1 Grant 2019-T1-002-015, and the Lee Kuan Yew Postdoctoral Fellowship at Nanyang Technological University Singapore. Any opinions, findings and conclusions or recommendations expressed in this material are those of the author(s) and do not reflect the views of National Research Foundation, Singapore. N.S. was supported by JSPS Grants-in-Aid for Scientific Research Grant Number JP19K14615.

\end{acknowledgments}

\bibliographystyle{apsrmp4-2}
\bibliography{myref}

% %%%%%%%%%%%%%%%%%%%%%%%%%%%%%%%%%%%%%%%%%%%%%%%%%%%%%%%%%%%%%%%%%%%%%%%%%%%%%%%%%%%%%%%%%%%%%%%%%%%%%%%%%%%%%%%%%%%%%%%%%%%%%%%%%%%%%%%%%%%%%%%%%%%%%%%%%%%%%%%%%%%%%%%%%%%%%%%%%%%%%%%%%%%%%%%%%%%%%%%%%%%%%%%%%%%%%%%%%%%%

%Comment out below for the word count. 

\appendix
\widetext

\setcounter{equation}{0}
\renewcommand{\theequation}{S.\arabic{equation}}

\section{Partial reusability of marginal catalysts}

One of the central motivations of catalytic resource transformation is tied to its reusability, i.e., the final state of the catalyst can be used for the next round of operation without degradation in the performance.   
Different classes of catalytic transformations come with different classes of reusability. 
If $\tau_C$ is a uncorrelated catalyst that admits a transformation 
$\mE(\rho\otimes\tau_C)=\rho'\otimes\tau_C$, we can reuse the final catalyst $\tau_C$ to run another round of the same transformation with input $\rho$. 
Here, the input state $\rho$ for the second round can either be freshly prepared or be prepared by applying another quantum channel $\Lambda$ to the final state as $\Lambda(\rho')=\rho$; the latter scenario is relevant to a thermal engine that makes a cyclic operation. 

A correlated catalyst $\tau_C$ that allows $\mE(\rho\otimes\tau_C)=\tilde \rho_{S'C}$ with $\Tr_C\tilde \rho_{S'C}=\rho'$ and $\Tr_{S'} \tilde \rho_{S'C}=\tau_C$ has a similar reusability property, although in a restricted form. 
After the initial transformation, one can apply $\mE$ over a freshly prepared input $\rho$ and the catalytic part of the final state $\tilde\rho_{S'C}$ to create another copy of the target state $\rho'$. Here, it is essential to use a fresh input $\rho$ that is uncorrelated with $\tilde \rho_{S'C}$ for the second round; in general, even if a channel $\Lambda$ can bring $\rho'$ back to $\rho$, one cannot use $\Lambda\otimes\id(\tilde \rho_{S'C})$ as an input to $\mE$ for the next round due to the remaining correlation between the main and catalytic systems.

We now see that marginal catalysts are also equipped with a certain class of reusability, which is naturally more restricted than that for correlated catalysts. 
Let $\otimes_{i=0}^{K-1}\tau_{C^{(i)}}$ be the marginal catalysts that enable the transformation $\mE\left(\rho\otimes \tau_{C^{(0)}}\cdots\otimes\tau_{C^{(K-1)}}\right)=\rho'\otimes\tau_{C^{(0)}\dots C^{(K-1)}}$ such that $\Tr_{\overline C^{(i)}}\tau_{C^{(0)}\dots C^{(K-1)}}=\tau_{C^{(i)}},\ \forall i$. 
Due to the correlation among catalytic subsystems, the total catalyst changes from the initial form and thus cannot be reused to the next operation even with a freshly prepared initial state $\rho$.
However, one can still utilize the property that the marginal state in each catalytic subsystem is intact, allowing one to reuse some parts of the marginal catalysts to realize the next transformation with the same performance. 
This \emph{partial reusability} admits significantly more transformations than the cases when no reusability can be exploited.

To show this explicitly, suppose that we are initially given $K$ copies of the marginal catalysts $\otimes_{i=0}^{K-1}\tau_{C^{(K-1)}}$. 
Each copy can be used to create a target state $\rho'$ with freshly prepared initial state $\rho$, realizing $K$ transformations in total.   
After the $K$ transformations, the property of marginal catalysts admits additional $K$ transformations --- each transformation reuses $K$ catalysts taken from different copies that are not correlated with each other. 
Fig.~\ref{fig:partial_reusability_small} shows an example of $K=3$; the left figure shows the initial three transformations and the right figure shows the following three transformations reusing the marginal catalysts. (The extension to an arbitrary $K$ is straightforward.) 
One might worry that after the fourth transformation (blue line in the right figure), $\tau_{C^{(1)}}$ in the first row and $\tau_{C^{(2)}}$ in the second row would get correlated due to the correlation between $\tau_{C^{(0)}}$ and $\tau_{C^{(1)}}$ in the first row generated by the first transformation (first row in the left figure), the correlation between $\tau_{C^{(1)}}$ and $\tau_{C^{(2)}}$ in the second row generated by the second transformation (second row in the left figure), and the correlation between $\tau_{C^{(0)}}$ in the first row and $\tau_{C^{(1)}}$ in the second row generated by the fourth transformation (blue line in the right figure); if this was the case, then the sixth transformation (green line in the right figure) would fail because of the correlation between $\tau_{C^{(1)}}$ in the first row and $\tau_{C^{(2)}}$ in the second row. 
However, a careful analysis shows that $\tau_{C^{(1)}}$ in the first row and $\tau_{C^{(2)}}$ in the second row are indeed uncorrelated even after the fourth transformation.
To see this, observe that the fourth transformation (blue line in the right figure) does not involve either $\tau_{C^{(1)}}$ in the first row or $\tau_{C^{(2)}}$ in the second row, and thus the marginal states on these systems do not get affected by this transformation. 
Since these marginal states are clearly uncorrelated just before the fourth transformation (after the initial three transformations), they remain uncorrelated after the fourth transformation. 
This analysis extends to other marginal states, showing the validity of the last three transformations.

This construction can be extended to scenarios in which more copies are provided. 
It is particularly insightful to consider the case when $K^n$ copies of $\otimes_{i=0}^{K-1}\tau_{C^{(K-1)}}$ are available, where $n$ is an arbitrary integer with $n\geq 2$.   
For instance, suppose $K^2$ copies are initially given (Fig.~\ref{fig:partial_reusability_large}). 
By forming $K$ groups of $K$ copies and making use of the above construction for each group, one can first run $K\times 2K = 2K^2$ transformations.  
Then, additional $K^2$ transformations can be realized by combining marginal catalysts from different groups; as shown in the left figure of Fig.~\ref{fig:partial_reusability_large}, the combinations of the first copy in each group can make $K$ transformations by extending the construction in Fig.~\ref{fig:partial_reusability_small}, and the same procedure can be applied to other copies (the middle and right figures of Fig.~\ref{fig:partial_reusability_large}).
Therefore, $2K^2+K^2 = 3K^2$ transformations are possible for $n=2$.
When $n= 3$, one can first form $K$ groups of $K^2$ copies, run $3K^2$ transformations for each group ($K\times 3K^2=3K^3$ in total) following the protocol for $n=2$, and make additional $K^2\times K= K^3$ transformations by combining marginal catalysts from different groups. (One group now has $K^2$ copies, each of which can run $K$ transformations together with the copies in the other groups.)
This totals $3K^3 + K^3 = 4 K^3$ transformations for $n=3$.
This procedure can be extended to an arbitrary $n$, allowing $(n+1) K^n$ transformations for $K^n$ copies of the marginal catalysts initially available.
This means that $(n+1)K^n - K^n = nK^n$ transformations are enabled by reusing the catalysts exploiting the property of the marginal catalysts that the marginal states remain invariant after the catalytic transformation, showing an operational value of the partial reusability of marginal catalysts. 

We also remark that this construction is one of many possible ways to exploit the partial reusability. 
Finding its optimal efficiency will be an interesting problem to study in a future work.

\begin{figure}
    \centering
    \includegraphics[scale=0.6]{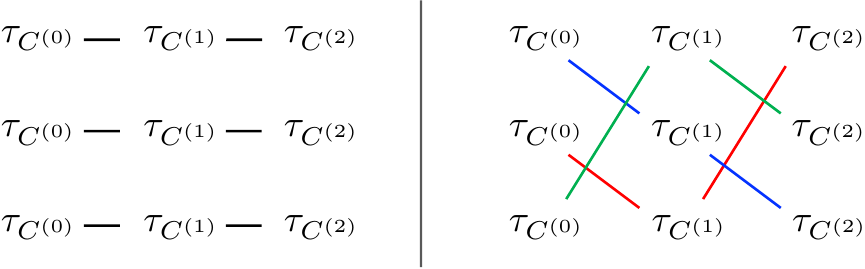}
    \caption{An example of $K=3$. (Left) Each copy of $\tau_{C^{(0)}}\otimes\tau_{C^{(1)}}\otimes\tau_{C^{(2)}}$ in the same row is used to run the first three rounds of the transformation, generating correlation within each copy. (Right) Additional three transformations can be made by using three combinations of $\tau_{C^{(0)}}$, $\tau_{C^{(1)}}$, and $\tau_{C^{(2)}}$ (connected by blue, red, and green lines) that are in different rows and not correlated with each other.}
    \label{fig:partial_reusability_small}
\end{figure}
\begin{figure}
    \includegraphics[width=0.7\columnwidth]{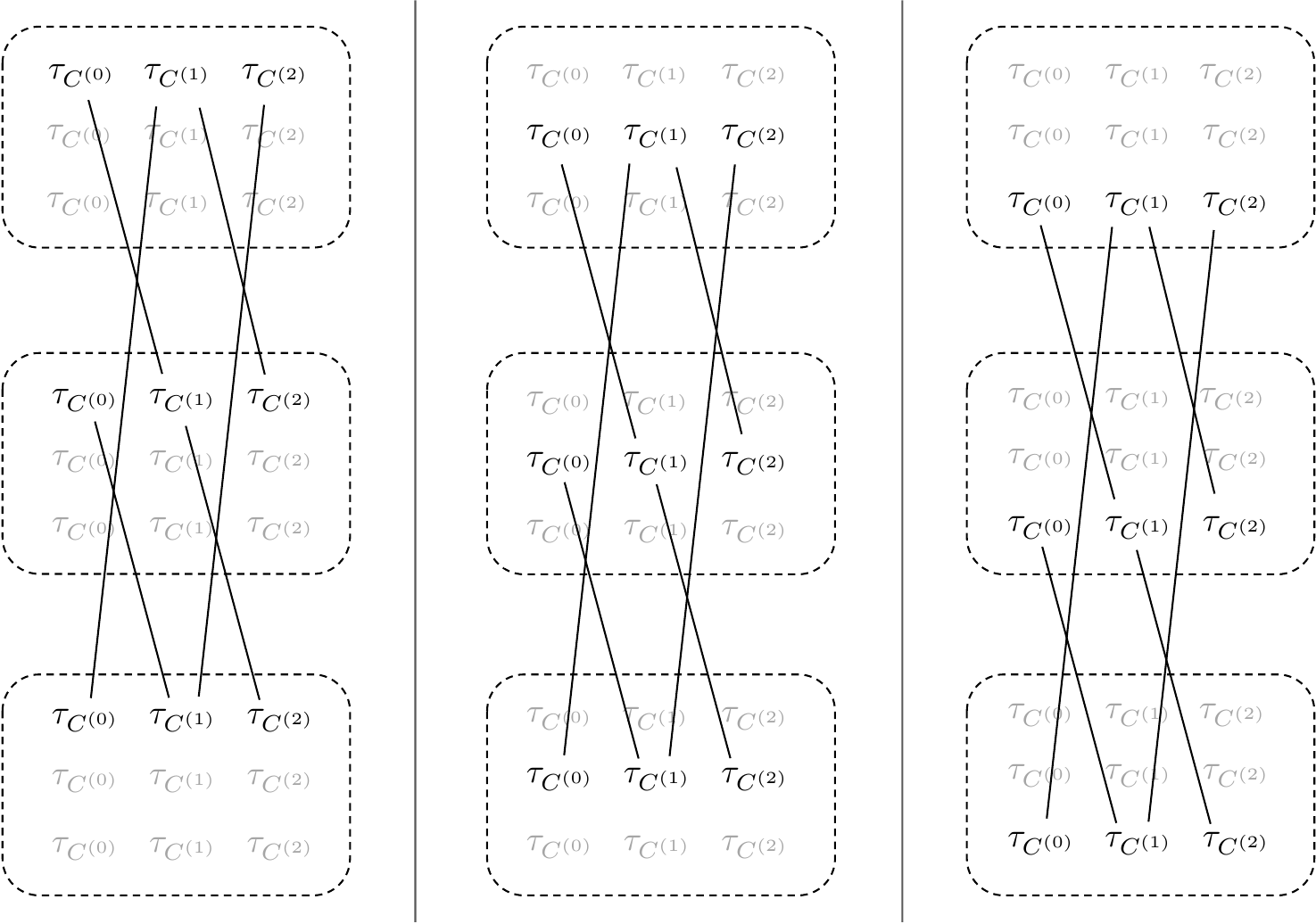}
    \caption{An example of $K=3$ and $n=2$. Each dotted box contains three copies of $\tau_{C^{(0)}}\otimes\tau_{C^{(1)}}\otimes\tau_{C^{(2)}}$ and runs six transformations as in Fig.~\ref{fig:partial_reusability_small}. (Left) Additional three transformations can be made by reusing three combinations of $\tau_{C^{(0)}}$, $\tau_{C^{(1)}}$, and $\tau_{C^{(2)}}$ that are in different boxes and not correlated with each other. Each combination is represented by three marginal states connected by black lines. 
    (Middle and Right) The same procedure can be repeated for different rows, enabling additional six transformations. 
    This protocol allows $3\cdot 3^2 = 27$ transformations in total, which can be compared to $3^2 = 9$ runs that are realized without employing the partial reusability of marginal catalysts.}
    \label{fig:partial_reusability_large}
\end{figure}

%%%%%%%%%%%%%%%%%%%%%%%%%%%%%%%%%%%%%%%%%%%%%%%%%%%%%%%%%%%%%%%%%%%%%%%%%%%%%%%%%%%%%%%%%%%%%%%%%%%%%%%%%%%%%%%%%%%%%%%%%%%%%%%

\section{Two-level coherence amplification}
Here, we review the protocol proposed in Ref.~\cite{Ding2021amplifying} which amplifies the coherence in two-level systems via marginal-catalytic covariant operations.
Although the authors claim that their protocol can keep the whole catalytic state intact with the presence of correlation only between the system and the catalytic system, unfortunately, their argument does not directly stand as we shall point out.   
Nevertheless, the protocol works for the marginal-catalytic transformation, which is the main focus of the current work.

Let $S$ and $C$ be two-level systems equipped with the same Hamiltonian $H_S=H_C=E\dm{1}$ with some energy $E$.
Let $\Sigma(\eta)\in\mD(S)$ be an initial state and $\Gamma(\eta)\in\mD(C)$ be a catalyst defined as
\bal
 \Sigma(\eta)&\coloneqq \frac{\mbI + \eta X}{2}\\
 \Gamma(\eta)&\coloneqq \frac{1}{2}\left(\mbI + \frac{\sqrt{3}\eta}{2}X + \frac{4-\eta^2}{6}Z\right)
\label{eq:states two-level amplification}
\eal
for $0<\eta<1$ where
$X\coloneqq \ketbra{0}{1}+\ketbra{1}{0}$ and $Z\coloneqq \dm{0}-\dm{1}$ are Pauli matrices. 
Consider a channel $\mE$ defined by the Kraus operators 
\bal
 K_0\coloneqq \begin{pmatrix}
  1 & 0& 0& 0 \\
  0 & \frac{1}{4}& \frac{\sqrt{3}}{4}& 0 \\
  0 & \frac{\sqrt{3}}{4}& \frac{3}{4}& 0 \\
  0 & 0& 0& 1 
  \end{pmatrix},\ \
 K_1\coloneqq  \begin{pmatrix}
  0 & 0& 0& 0 \\
  0 & 0& 0& 0 \\
  0 & -\frac{\sqrt{3}}{2}& \frac{1}{2}& 0 \\
  0 & 0& 0& 0  
 \end{pmatrix}.
  \label{eq:two-level amplification Kraus}
\eal

One can check that $\mE$ is a covariant operation by observing that the first Kraus operator nontrivially acts only on the degenerate subspace spanned by $\ket{01}$ and $\ket{10}$,
and the second Kraus operator acts as a projection from a state in the degenerate subspace to an incoherent state.
A direct calculation reveals that 
\bal
 \Tr_C\mE\left(\Sigma(\eta)\otimes\Gamma(\eta)\right)=&\Sigma(\eta'), \\ \Tr_S\mE\left(\Sigma(\eta)\otimes\Gamma(\eta)\right)=&\Gamma(\eta)
 \label{eq:two-level amplification}
\eal
where $\eta'\coloneqq\frac{\eta(25-\eta^2)}{24}$.
Since $\eta'>\eta$ for $0<\eta<1$, the resultant state in $S$ gets closer to $\dm{+}=(\mbI + X)/2$ than the initial state $\rho$. 
We sequentially apply this protocol until we obtain a state close to $\ket{+}$ within desired accuracy. 
More specifically, suppose that the initial state in $S$ is $\Sigma(\eta_0)$ for some $0<\eta_0<1$. 
We inductively define the sequence $\{\eta_j\}_j$ by
\bal
\eta_{j+1}\coloneqq\frac{\eta_j(25-{\eta_j}^2)}{24}.
\label{eq:catalyst family constants}
\eal
Let $\Gamma(\eta_j)$ be a catalyst defined in $C^{(j)}$ and $\mE_{SC^{(j)}}$ be the operation acting on $SC^{(j)}$ defined by the Kraus operators in \eqref{eq:two-level amplification Kraus}.
We denote the state prepared by the $K$ sequential application of the amplification protocol by 
\bal
 \tilde\rho_{SC^{(0)}\dots C^{(K-1)}}\coloneqq\mE_{SC^{(K-1)}}\circ\dots\circ\mE_{SC^{(1)}}\circ\mE_{SC^{(0)}}\left(\Sigma(\eta_0)\otimes\Gamma(\eta_0)\otimes \Gamma(\eta_1)\otimes\dots\otimes \Gamma(\eta_{K-1})\right).
 \label{eq:two-level amplification protocol}
\eal
Owing to \eqref{eq:two-level amplification}, we get
$\Tr_{\overline{S}}\tilde\rho_{SC^{(0)}\dots C^{(K-1)}}=\Sigma(\eta_K)$ and $\Tr_{\overline{C^{(j)}}}\tilde\rho_{SC^{(0)}\dots C^{(K-1)}}=\Gamma(\eta_j)\ \forall j$. 
By taking $K$ sufficiently large, one can bring $\Sigma(\eta_K)$ arbitrarily close to $\ket{+}$.

Note that we use a slightly different Kraus operator in \eqref{eq:two-level amplification Kraus} from the one introduced in Ref.~\cite{Ding2021amplifying}. The advantage of using the form in \eqref{eq:two-level amplification Kraus} is that the resultant state in $S$ after an application of $\mE$ ends up in the family $\{\Sigma(\eta)\}_\eta$, which can be directly used for the next round.
On the other hand, the channel considered in Ref.~\cite{Ding2021amplifying} induces an additional $Z$ term in the resultant state that needs to be addressed by introducing another ancillary system in the catalytic system to make the sequential application work. 

\bigskip

The transformation from $\Sigma(\eta_0)$ to $\Sigma(\eta_K)$ should be considered as a marginal-catalytic covariant transformation rather than a correlated-catalytic covariant transformation because correlation can be generated among catalytic systems $C^{(0)}\dots C^{(K-1)}$. 
The authors of Ref.~\cite{Ding2021amplifying} claim that one can turn this protocol into a correlated-catalytic covariant transformation by using a catalyst defined in $C^{(0)}\dots C^{(K-1)}$ that already has correlation between the subsystems; namely, defining $\mE_{0:K-1}\coloneqq\mE_{SC^{(K-1)}}\circ\dots\circ\mE_{SC^{(1)}}\circ\mE_{SC^{(0)}}$, they propose a catalyst $\tau_{C^{(0)}\dots C^{(K-1)}}$ such that $\Tr_{\overline{S}}\mE_{0:K-1}(\Sigma(\eta_0)\otimes\tau_{C^{(0)}\dots C^{(K-1)}})=\Sigma(\eta_K)$ and $\Tr_{S}\mE_{0:K-1}(\Sigma(\eta_0)\otimes\tau_{C^{(0)}\dots C^{(K-1)}})=\tau_{C^{(0)}\dots C^{(K-1)}}$. 
Their argument is based on their observation that any state $\tilde\tau_{C^{(0)}\dots C^{(K-1)}}$ such that $\Tr_{\overline{C^{(j)}}}\tilde\tau_{C^{(0)}\dots C^{(K-1)}}=\Gamma(\eta_j)\ \forall j$ would allow one to run the protocol in \eqref{eq:two-level amplification protocol} with the same performance at the level of marginal states, i.e., 
\bal
 \Tr_{\overline{S}}\mE_{0:K-1}(\Sigma(\eta_0)\otimes \tilde\tau_{C^{(0)}\dots C^{(K-1)}})=\Sigma(\eta_K)
 \label{eq:Ding claim 1}
\eal
and
\bal
 \Tr_{\overline{C^{(j)}}}\mE_{0:K-1}(\Sigma(\eta_0)\otimes \tilde\tau_{C^{(0)}\dots C^{(K-1)}})=\Gamma(\eta_j),
 \label{eq:Ding claim 2}
\eal
because each $\mE_{SC^{(k)}}$ does not ``touch'' $C^{(k')}$ with $k'\neq k$. 
However, unfortunately, Eqs.~\eqref{eq:Ding claim 1}, \eqref{eq:Ding claim 2} do not hold in general.
Even though $\mE_{SC^{(k)}}$ only acts on $SC^{(k)}$, it \emph{can} create correlation between $S$ and $C_{k'}$ for $k'\neq k$ through the initial correlation contained in $\tilde\tau_{C^{(0)}\dots C^{(K-1)}}$, preventing one from using the relations in \eqref{eq:two-level amplification}.
To see this more explicitly, let us consider the case $L=2$. (Fig.~\ref{fig:twolevel}).
We first run the protocol starting with the product catalysts $\Gamma(\eta_0)\otimes\Gamma(\eta_1)\in\mD(C^{(0)}\otimes C^{(1)})$. 
The resulting state in $C^{(0)}C^{(1)}$, $\tilde\tau_{C^{(0)}C^{(1)}}\coloneqq \Tr_S\mE_{SC^{(1)}}\circ\mE_{SC^{(0)}}(\Sigma(\eta_0)\otimes\Gamma(\eta_0)\otimes\Gamma(\eta_1))$, keeps the original states as its marginals as $\Tr_{\overline{C^{(0)}}}\tilde\tau_{C^{(0)}C^{(1)}}=\Gamma(\eta_0)$, $\Tr_{\overline{C^{(1)}}}\tilde\tau_{C^{(0)}C^{(1)}}=\Gamma(\eta_1)$ while $\tilde\tau_{C^{(0)}C^{(1)}}\neq \Gamma(\eta_0)\otimes\Gamma(\eta_1)$ due to the correlation generated via $S$.  
Now, let us use $\tilde\tau_{C^{(0)}C^{(1)}}$ as a catalyst state to start with. 
Let $\tilde \rho_{SC^{(1)}}\coloneqq\Tr_{C^{(0)}}\mE_{SC^{(0)}}(\Sigma(\eta_0)\otimes\tilde\tau_{C^{(0)}C^{(1)}})$ be the state on $SC^{(1)}$ after the application of $\mE_{SC^{(0)}}$. 
Then, although it holds that $\Tr_{C^{(1)}}\tilde\rho_{SC^{(1)}}=\Sigma(\eta_1)$ and $\Tr_{S}\tilde\rho_{SC^{(1)}}=\Gamma(\eta_1)$, it is no longer a product state, i.e., $\tilde\rho_{SC^{(1)}}\neq \Sigma(\eta_1)\otimes\Gamma(\eta_1)$. 
Therefore, one cannot rely on \eqref{eq:two-level amplification} anymore, and in particular, we get
\bal
  \Tr_{\overline{C^{(1)}}}\mE_{0:1}(\Sigma(\eta_0)\otimes\tilde\tau_{C^{(0)}C^{(1)}})\neq \Gamma(\eta_1),
\eal
violating \eqref{eq:Ding claim 2}.
This can be explicitly checked by a direct calculation; the specific form for $\Tr_{\overline{C^{(1)}}}\mE_{0:1}(\Sigma(\eta_0)\otimes\tilde\tau_{C^{(0)}C^{(1)}})$ is obtained as

\bal
\begin{pmatrix}
 \frac{11919015936 - 1140507536 \eta_0^2 + 91814899 \eta_0^4 + 
 1471879 \eta_0^6 - 161703 \eta_0^8 + 2493 \eta_0^{10}}{14495514624} && \frac{\eta_0 (412672 - 19231 \eta_0^2 - 234 \eta_0^4 + 
  9 \eta_0^6)}{524288 \sqrt{3}}\\
  \frac{\eta_0 (412672 - 19231 \eta_0^2 - 234 \eta_0^4 + 
  9 \eta_0^6)}{524288 \sqrt{3}} && \frac{2576498688 + 1140507536 \eta_0^2 - 91814899 \eta_0^4 - 
 1471879 \eta_0^6 + 161703 \eta_0^8 - 2493 \eta_0^{10}}{14495514624}
 \end{pmatrix}
\eal
whereas
\bal
 \Gamma(\eta_1) = \begin{pmatrix}
  \frac{1}{12}\left(10 - \frac{\eta_0^2}{576} (25 - \eta_0^2)^2\right) && \frac{\eta_0 (25 - \eta_0^2)}{32 \sqrt{3}}\\
  \frac{\eta_0 (25 - \eta_0^2)}{32 \sqrt{3}}  && \frac{1152 + 625 \eta_0^2 - 50 \eta_0^4 + \eta_0^6}{6912}
 \end{pmatrix}.
\eal

\begin{figure}
    \centering
    \includegraphics[scale=0.5]{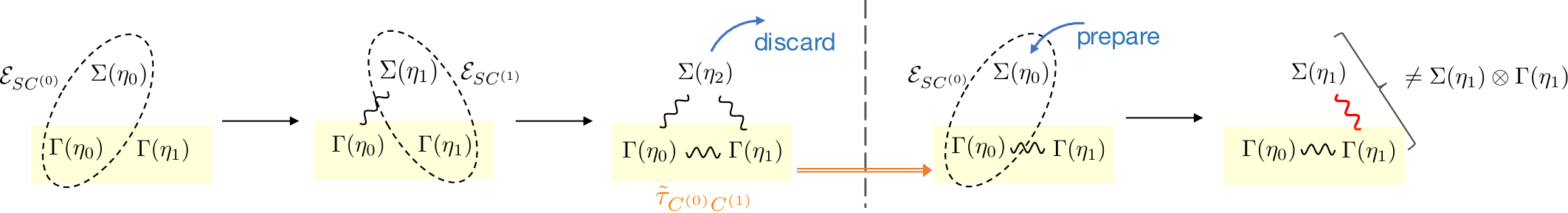}
    \caption{We first apply $\mE_{SC^{(0)}}$ over $\Sigma(\eta_0)\otimes\Gamma(\eta_0)$ and $\mE_{SC^{(1)}}$ over $\Sigma(\eta_1)\otimes\Gamma(\eta_1)$ sequentially, which produces reduced state $\tilde\tau_{C^{(0)}C^{(1)}}$ in the catalytic system. Then, we prepare another state $\Sigma(\eta_0)$ and apply $\mE_{SC^{(0)}}$ over $\Sigma(\eta_0)$ and $\tilde\tau_{C^{(0)}C^{(1)}}$. Even though $\mE_{SC^{(0)}}$ does not ``touch'' $C^{(1)}$, it can generate correlation between $S$ and $C^{(1)}$ via the initial correlation between $C^{(0)}$ and $C^{(1)}$ in $\tilde\tau_{C^{(0)}C^{(1)}}$.}
    \label{fig:twolevel}
\end{figure}
%%%%%%%%%%%%%%%%%%%%%%%%%%%%%%%%%%%%%%%%%%%%%%%%%%%%%%%%%%%%%%%%%%%%%%%%%%%%%%%%%%%%%%%%%%%%%%%%%%%%%%%%%%%%%%

\section{On Theorem~\ref{thm:transform catalyst}}

We first give a detailed proof of Theorem~\ref{thm:transform catalyst} and later comment on the extension of our protocol to implement the desired quantum channel and the size of catalysts used in our protocol. 

\vspace{\baselineskip}

{\it Theorem~\ref{thm:transform catalyst}}:
For any $\rho\in\mD(S)$, $\rho'\in\mD(S')$ and $\epsilon>0$, $\rho$ can be transformed to a state $\rho'_\epsilon\in\mD(S')$ such that $\frac{1}{2}\|\rho'-\rho'_\epsilon\|_1\leq \epsilon$ by a marginal-catalytic covariant transformation.

\bigskip

\begin{proof}
Since covariant operations and marginal-catalytic covariant operations are closed under concatenation, it suffices to check that every operation in the protocol is either covariant or marginal-catalytic covariant.
It is useful to also recall that taking the partial trace over any subsystems and preparing incoherent states with respect to the Hamiltonian of the added system are both covariant operations. 
We provide detailed descriptions for each step in the protocol. 

\textit{\textbf{Step 1}}---
Since we will construct an arbitrary state from scratch, we first trace out the initial state in $S$.
Next, we prepare $\ket{0}_{R_i}$ for $i\in\{0,\dots,d_{S'}-1\}$, the ground state of two-level system $R_i$ with Hamiltonian
\bal
 H_{R_i} = (E_{S',i}-E_{S',j^\star}) \dm{1}_{R_i}
\label{eq:hamiltonian reference}
\eal
where $\{E_{S',j}\}_{i=0}^{d_{S'}-1}$ is the eigenenergy of the Hamiltonian in the final system $S'$, and $j^\star$ is an arbitrary fixed index independent of $i$.
Recall that the preparation of $\ket{0}_{R_i}$ is covariant.

Let $\eta$ be an arbitrary real number satisfying $0<\eta<1$. 
We introduce catalysts $\tau_i^{a}\coloneqq\Sigma(\eta)$ and $\tau_i^{b}\coloneqq\Gamma(\eta)$ as in \eqref{eq:states two-level amplification} in catalytic subsystems $C_i^{a}$ and $C_i^{b}$ respectively, where both subsystems are equipped with Hamiltonian $H_{R_i}$.
We shall first amplify the coherence in $C_i^a$ by using $C_i^b$ as a correlated catalyst, and then distribute the increment of coherence to the main subsystem $R_i$.
We first apply the channel with the Kraus operators in \eqref{eq:two-level amplification Kraus} on $\tau_i^{a}\otimes\tau_i^{b}$ to obtain $\Sigma(\eta')$ with $\eta'=\frac{\eta(25-\eta^2)}{24}$ as a reduced state on $C_i^{a}$ while retaining $\Gamma(\eta)$ as a reduced state on $C_i^{b}$.
We next distribute coherence to $R_i$.
Let $U_{R_iC_i^{a}}$ be a unitary matrix acting on $R_iC_i^{a}$ defined as 
\bal
U_{R_iC_i^{a}}:=\dm{00}+\dm{11}+\left(\alpha\ketbra{01}{01}+\sqrt{1-\alpha^2}\ketbra{01}{10}-\sqrt{1-\alpha^2}\ketbra{10}{01}+\alpha\ketbra{10}{10}\right)
\label{eq:unitary to create seed}
\eal
where $\alpha$ is a real number determined shortly. 
Since this preserves the total energy, the application of this unitary is covariant. 
However, because of the nontrivial transformation within the degenerate subspace ${\rm span}\{\ket{01}, \ket{10}\}$, non-zero coherence can be `transferred' from $C_i^{a}$ to $R_i$. 
Let $\xi_i\coloneqq \Tr_{C_i^{a}}\left[U_{R_iC_i^{a}}(\dm{0}_{R_i}\otimes \Sigma(\eta')) U_{R_iC_i^{a}}^\dagger\right]$ and $\tilde\tau_i^{a}\coloneqq \Tr_{R_i}\left[U_{R_iC_i^{a}}(\dm{0}_{R_i}\otimes \Sigma(\eta')) U_{R_iC_i^{a}}^\dagger\right]$ be the reduced final state of the main subsystem $R_i$ and the catalytic subsystem $C_i^a$, respectively.
By direct calculation, one can check that
\bal
 \xi_i &= \frac{1}{2}\begin{pmatrix}
 1+\alpha^2 && -\eta'\sqrt{1-\alpha^2}\\ -\eta'\sqrt{1-\alpha^2} && 1-\alpha^2
\end{pmatrix},\\ 
\tilde\tau_i^{a}&=\frac{1}{2}\begin{pmatrix}
2-\alpha^2 && a \eta' \\ a \eta' && \alpha^2
\end{pmatrix}.
\eal
To restore the marginal state of $C_i^a$ to its original state, we apply a covariant operation defined by Kraus operators $K_0=\frac{1}{\sqrt{2-\alpha^2}}\dm{0}+\dm{1}$ and $K_1=\sqrt{\frac{1-\alpha^2}{2-\alpha^2}}\ketbra{1}{0}$ to $\tilde\tau_i^{a}$ and get $\Sigma(\tilde\eta)$ with $\tilde\eta = \frac{\alpha\eta'}{\sqrt{2-\alpha^2}}$.
Since we would like to have $\tilde\eta=\eta$ to restore $\tau_i^{a}$ on $C_i^{a}$, we choose $\alpha$ so that 
\bal
 \eta = \frac{\alpha\eta(25-\eta^2)}{24\sqrt{2-\alpha^2}},\ \ (0<\alpha<1),
\eal
where a direct calculation gives 
\bal
 \alpha = \sqrt{2\left[1+\left(\frac{25-\eta^2}{24}\right)^2\right]^{-1}}.
\eal
We note that $\frac{25-\eta^2}{24}>1$ for $0<\eta<1$ confirms the condition $\alpha<1$.
On the other hand, applying a covariant operation with Kraus operators $K_0=\frac{1}{\sqrt{1+\alpha^2}}\dm{0}-\dm{1}, K_1=\frac{\alpha}{\sqrt{1+\alpha^2}}\ketbra{1}{0}$ on $R_i$ brings $\xi_i$ to $\Sigma\left(\frac{\eta'\alpha\sqrt{1-\alpha^2}}{\sqrt{1+\alpha^2}}\right)$, which now has a non-zero coherence.

\textit{\textbf{Step 2}}---
To amplify the coherence generated in $R_i$, we introduce $K$ catalytic subsystems $C_i^{(0)},C_i^{(1)},\dots,C_i^{(K-1)}$ where each subsystem is equipped with the Hamiltonian $H_{R_i}$. We prepare catalysts $\otimes_{j=0}^{K-1}\Gamma(\eta_j)$ in $\otimes_{j=0}^{K-1}C_i^{(j)}$, where we set $\eta_0=\frac{\eta'\alpha\sqrt{1-\alpha^2}}{\sqrt{1+\alpha^2}}$ and choose $\eta_j$ for $j=1,2,\dots,K-1$ inductively by \eqref{eq:catalyst family constants}.  
The two-level coherence amplification protocol gives $\Sigma(\eta_{K})=\frac{1}{2}(\mbI+\eta_{K}X)$ in $R_i$ by a marginal-catalytic covariant operation.
In particular, a state arbitrarily close to $\ket{+}_{R_i}$ can be prepared by taking sufficiently large $K$. 

\textit{\textbf{Step 3}}---
We repeat Step~1 and 2 for $L$ times to prepare a state close to $\ket{+}_{R_i}^{\otimes L}$ for each $i\in\{0,\dots,d_{S'}-1\}$.
How close between the prepared state and $\ket{+}_{R_i}^{\otimes L}$ is determined by $K$ chosen above. 
In system $S'$, we prepare $\ket{j^\star}_{S'}$ where $j^\star$ is the index introduced in \eqref{eq:hamiltonian reference}, and $\ket{j^\star}_{S'}$ is an eigenstate of $H_{S'}$ with eigenvalue $E_{S',j^\star}$, which is an incoherent state and thus can be prepared by a covariant operation. 
We use the state in $R=\otimes_{i=0}^{d_{S'}-1} R_i^{\otimes L}$ as a coherent resource to implement an operation on $S'$ that prepares a state close to the target state $\rho'$.
We first note that for each $R_i$, the resource state can be written as 
\bal
 \ket{+}_{R_i}^{\otimes L} &= \frac{1}{2^{L/2}}\sum_{n=0}^{L} \binom{L}{n}^{1/2}\ket{n},\\ \ket{n}&\coloneqq\binom{L}{n}^{-1/2}\sum_{\substack{{\bf i}\in \{0,1\}^{L}\\|{\bf i}|=n}}\ket{{\bf i}}
 \label{eq:resource binomial}
\eal
where $|{\bf v}|$ is the number of $1$'s in the bit string ${\bf v}$, and $\binom{m}{n}\coloneqq \frac{m!}{n!(m-n)!}$ is a binomial coefficient.
By construction, $\ket{n}$ is an energy eigenstate of the Hamiltonian of the total systems $R_i^{\otimes L}$ defined as $H_{R_i}^{\otimes L}=H_{R_i}\otimes \mbI_{R_i^{\otimes L-1}} + \mbI_{R_i}\otimes H_{R_i}\otimes \mbI_{R_i^{\otimes L-2}}+\dots + \mbI_{R_i^{\otimes L-1}}\otimes H_{R_i}$ with energy $n(E_{S',i}-E_{S',j^\star})$.
For later use, we embed the state on $R_i^{\otimes L}$ into a system with one additional dimension $\tilde R_i^{L}$ with Hamiltonian 
\bal
H_{\tilde R_i^{L}}=H_{R_i^{\otimes L}}-(E_{S',i}-E_{S',j^\star})\dm{-1},
\eal
where $\ket{-1}$ is an eigenstate of the new Hamiltonian orthogonal to all the eigenstates of $H_{R_i^{\otimes L}}$. 
The embedding $\sum_i \ket{i}_{\tilde R_i^L}\bra{i}_{R_i^{\otimes L}}$ is a covariant operation.  
By restricting our attention to $\{ \ket{n}\}_{n=-1}^L$, this system can be regarded as a ladder system where the state $\ket{n}$ has its energy $n(E_{S',i}-E_{S',j^\star})$, and in this view the embedding of $\ket{-1}$ can be understood as the extension of the ladder to one level below.
This state $\ket{+}_{R_i}^{\otimes L}$ serves as a coherent resource state in the approximation of unitary operations.

For our aim, it suffices to show that any pure state $\psi'$ in $S'$ can be created because any mixed state $\rho'$ can be obtained by stochastically preparing the pure states that appear in a convex combination constituting $\rho'$.
For any pure state $\psi'$, there exists a unitary $V$ such that
\bal
\ket{\psi'}=V\ket{j^\star}_{S'}.
\label{eq:goal pure}
\eal 
Then, we consider a channel with Kraus operators

\bal
 K_0 := \sum_{k=0}^{d_{S'}-1}\left[ \dm{k}V\dm{j^\star}\otimes\Delta_{\tilde R_k^L}\bigotimes_{j\neq k} P_{\tilde R_j^L}^{0:L} \right]
\eal
\bal
 K_1 := (\mbI_{S'}-\dm{j^\star})\bigotimes_{k=0}^{d_{S'}-1} P_{\tilde R_k^L}^{0:L}+\mbI_{S'}\otimes\left[\mbI-\bigotimes_{k=0}^{d_{S'}-1} P_{\tilde R_k^L}^{0:L}\right]
\eal
where $P_{\tilde R_k^L}^{0:L}\coloneqq \sum_{n=0}^{L}\dm{n}_{\tilde R_k^L}$ and $\Delta_{\tilde R_k^L}$ is the energy-shifting operator acting on system $\tilde R_k^L$ defined as 
\bal
 \Delta_{\tilde R_k^L} \coloneqq \sum_{n=0}^{L} \ketbra{n-1}{n}_{\tilde R_k^L},
 \label{eq:energy-shifting}
\eal
where $\ket{n}$ is the energy eigenstate defined in \eqref{eq:resource binomial}.
One can check that $K_0$ and $K_1$ are valid Kraus operators since

\bal
 K_0^\dagger K_0 &= \sum_{k=0}^{d_{S'}-1}  \dm{j^\star}V^\dagger\dm{k}V\dm{j^\star}\otimes \Delta_{\tilde R_k^L}^\dagger\Delta_{\tilde R_k^L}\bigotimes_{j\neq k}P_{\tilde R_j^L}^{0:L}\\
 &=\dm{j^\star}\bigotimes_{k=0}^{d_{S'}-1}P_{\tilde R_j^L}^{0:L},
\eal
where we used $\Delta_{\tilde R_k^L}^\dagger\Delta_{\tilde R_k^L}=P_{\tilde R_k^L}^{0:L}$, and
\bal
 K_1^\dagger K_1 =  (\mbI_{S'}-\dm{j^\star})\bigotimes_{k=0}^{d_{S'}-1} P_{\tilde R_k^L}^{0:L}+\mbI_{S'}\otimes\left[\mbI-\bigotimes_{k=0}^{d_{S'}-1} P_{\tilde R_k^L}^{0:L}\right]
\eal
satisfy $K_0^\dagger K_0 + K_1^\dagger K_1 = \mbI_{S'}\otimes\mbI_{\tilde R}$.
In addition, this channel is covariant since each term in $K_0$ conserves the total energy and thus commutes with the total Hamiltonian, and $K_1$ is just a projector onto a subspace spanned by the energy eigenbasis.
If we apply the above channel to $\dm{j^\star}\bigotimes_{k=0}^{d_{S'}-1}\dm{+}_{\tilde R_k^L}^{\otimes L}$ and take the partial trace over $\tilde R$ systems, $K_0$ deterministically clicks and gives  
\bal
\Tr_{\tilde R}&\left[K_0 \left( \dm{j^\star}\bigotimes_{k=0}^{d_{S'}-1}\dm{+}_{\tilde R_k^L}^{\otimes L} \right) K_0^\dagger\right] \\&= \sum_{k=0}^{d_{S'}-1} \sum_{l\neq k} \dm{k}V\dm{j^\star}V^\dagger\dm{l} \Tr\left[\Delta_{\tilde R_k^L}\dm{+}_{\tilde R_k^L}^{\otimes L}\right]\Tr\left[\dm{+}_{\tilde R_l^L}^{\otimes L}\Delta_{\tilde R_l^L}^\dagger\right]+\sum_{k=0}^{d_{S'}-1} \dm{k}V\dm{j^\star}V^\dagger\dm{k}
\label{eq:applying unitary}
\eal
where we used $\Tr\left[\Delta_{\tilde R_k^L}\dm{+}_{\tilde R_k^L}^{\otimes L}\Delta_{\tilde R_k^L}^\dagger\right]=1$ to get the second term. 
Since we can prepare a state arbitrarily close to $\bigotimes_{k=0}^{d_{S'}-1}\dm{+}_{\tilde R_k^L}^{\otimes L}$ by taking sufficiently large $K$ in Step~1 and $L$ in Step~2, the obtained state is close to \eqref{eq:applying unitary} with arbitrary accuracy. 
Furthermore, since $\Tr\left[\Delta_{\tilde R_k^L}\dm{+}_{\tilde R_k^L}^{\otimes L}\right]$ converges to 1 in the limit of $L\to\infty$ as we show below, the state in \eqref{eq:applying unitary} can be made arbitrarily close to $\psi'=V\dm{j^\star}V^\dagger$ by taking large enough $L$.

We shall show $\Tr\left[\Delta_{\tilde R_k^L}\dm{+}_{\tilde R_k^L}^{\otimes L}\right]\to 1$ as $L\to\infty$. 
Using the expression in \eqref{eq:resource binomial}, we have 
\bal
 \Tr\left[\Delta_{\tilde R_k^L}\dm{+}_{\tilde R_k^L}^{\otimes L}\right]=\frac{1}{2^L}\sum_{n=0}^{L-1}\binom{L}{n+1}^{1/2}\binom{L}{n}^{1/2}.
 \label{eq:shift overlap}
\eal
Let us first consider the case when $L$ is even.
Then, it holds that
\bal
\begin{cases}
 \binom{L}{n+1}\geq \binom{L}{n} & n\leq  \frac{L}{2}-1 \\
 \binom{L}{n}\geq \binom{L}{n+1} & n\geq  \frac{L}{2}
\end{cases}.
\eal
Thus, \eqref{eq:shift overlap} can be bounded as 
\bal
 \Tr\left[\Delta_{\tilde R_k^L}\dm{+}_{\tilde R_k^L}^{\otimes L}\right]
 &=\frac{1}{2^L}\sum_{n=0}^{L-1} \binom{L}{n}^{1/2}\binom{L}{n+1}^{1/2} \\
 &\geq \frac{1}{2^L}\sum_{n=0}^{\frac{L}{2}-1}\binom{L}{n}+\frac{1}{2^L}\sum_{n=\frac{L}{2}}^{L-1}\binom{L}{n+1}\\
 &=\frac{1}{2^L}\sum_{n=0}^{\frac{L}{2}-1}\binom{L}{n}+\frac{1}{2^L}\sum_{n=\frac{L}{2}+1}^{L}\binom{L}{n}\\
 &=1-\frac{1}{2^L}\binom{L}{L/2}.
 \label{eq:shift overlap bound}
\eal
Stirling's formula~\cite{olver2010nist} implies that for any $\delta>0$, there exists a sufficiently large $n$ such that $\left|n!-\sqrt{2\pi}n^{n+1/2}e^{-n}\right|<\delta$.
This provides a lower bound for the last line in \eqref{eq:shift overlap bound} which holds for any $\delta_e>0$ with sufficiently large $L$ as 
\bal
1- \frac{1}{2^L} \binom{L}{L/2}&\geq 1-\frac{1}{2^L}\frac{\sqrt{2\pi}L^{L+1/2}e^{-L}}{\left[\sqrt{2\pi}(L/2)^{L/2+1/2}e^{-L/2}\right]^2}-\delta_e\\
&=1-\frac{2}{\sqrt{2\pi}}L^{-1/2}-\delta_e\xrightarrow[L\to\infty]{} 1-\delta_e. 
\eal

The case when $L$ is odd can be bounded similarly. For any $\delta_o>0$ and sufficiently large $L$, it holds that
\bal
 \Tr\left[\Delta_{\tilde R_k^L}\dm{+}_{\tilde R_k^L}^{\otimes L}\right]
  &\geq \frac{1}{2^L}\sum_{n=0}^{\frac{L-1}{2}}\binom{L}{n}+\frac{1}{2^L}\sum_{n=\frac{L+1}{2}}^{L-1}\binom{L}{n+1}\\
 &=\frac{1}{2^L}\sum_{n=0}^{\frac{L-1}{2}}\binom{L}{n}+\frac{1}{2^L}\sum_{n=\frac{L+1}{2}+1}^{L}\binom{L}{n}\\
 &=1-\frac{1}{2^L}\binom{L}{\frac{L+1}{2}}\\
 &\geq  1-\frac{1}{2^L}\frac{\sqrt{2\pi}L^{L+1/2}e^{-L}}{\left[\sqrt{2\pi}\left(\frac{L+1}{2}\right)^{L/2+1}e^{-(L+1)/2}\right]\left[\sqrt{2\pi}\left(\frac{L-1}{2}\right)^{L/2}e^{-(L-1)/2}\right]}-\delta_o\\
 &\geq  1-\frac{2}{\sqrt{2\pi}}(1+1/L)^{-L/2}(1-1/L)^{-L/2}\left(\sqrt{L}+1/\sqrt{L}\right)^{-1}-\delta_o\xrightarrow[L\to\infty]{} 1-\delta_o.
 \label{eq:shift overlap bound odd}
\eal
Since $\delta_e,\delta_o>0$ can be taken arbitrarily small, we obtain $\lim_{L\to\infty}\Tr\left[\Delta_{\tilde R_k^L}\dm{+}_{\tilde R_k^L}^{\otimes L}\right] \geq 1$.
Together with $\Tr\left[\Delta_{\tilde R_k^L}\dm{+}_{\tilde R_k^L}^{\otimes L}\right] \leq 1$,  we get  $\lim_{L\to\infty}\Tr\left[\Delta_{\tilde R_k^L}\dm{+}_{\tilde R_k^L}^{\otimes L}\right] = 1$.

Finally, we can eliminate the correlation between $S'$ and $C$ by using catalysts that depend on the final state by using the technique introduced in Ref.~\cite{Muller2016generalization}. (See also Fig.~\ref{fig:remove_correlation}).
Let $\tau_{C^{(0)}},\dots,\tau_{C^{(N-1)}}$ be the catalysts that prepare $\rho'_\epsilon$ in the reduced state on $S'$. 
Let us take $\tau_{C^{(N)}}\coloneqq\rho'_\epsilon$. 
We first run the protocol using $\tau_{C^{(0)}},\dots,\tau_{C^{(N-1)}}$ to obtain a state whose marginal on $S'$ is $\rho'_\epsilon$ and finally swap $S'$ and $C^{(N)}$. 
This prepares $\rho'_\epsilon$ in $S'$ decoupled from the catalytic systems, and
since the SWAP operation between two isomorphic systems is covariant, the total operation is still covariant. 

\begin{figure}
    \centering
    \includegraphics[scale=0.5]{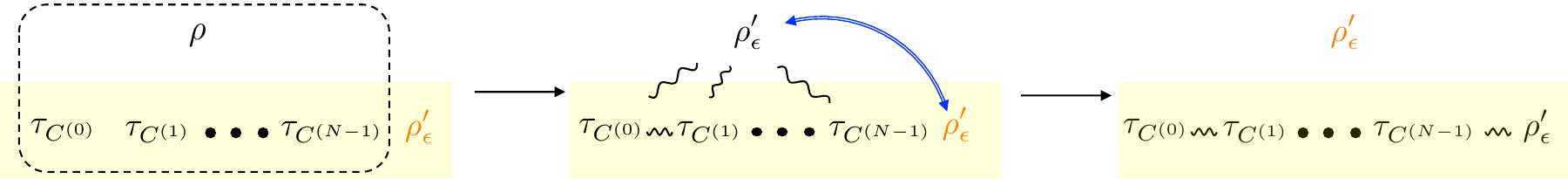}
    \caption{Besides the catalysts used for the main protocol, we also prepare the final state $\rho'_\epsilon$ in the catalytic system. After successfully creating a total state whose marginal on $S'$ is an approximate target state, we swap the marginal state and the copy prepared in the catalytic system. Since marginal states in the catalytic subsystems unchanged, this still constitutes a valid marginal-catalytic protocol.}
    \label{fig:remove_correlation}
\end{figure}

\end{proof}

\begin{remark}
Step~3 enables one to implement any unitary $V$ on the main system with arbitrary accuracy. This construction can be used to implement an arbitrary \emph{quantum channel} $\mE:S\rightarrow S_{\rm out}$ that acts on the initial state $\rho$ only using covariant operations assisted with marginal catalysts. 
Let $\mE(\cdot) = \Tr_{\overline{S}_{\rm out}}[V(\cdot\otimes\dm{0}_E)V^\dagger]$ be a Stinespring dilation for $\mE$, where $\ket{0}_E$ is a pure incoherent state on an ancillary system $E$, $V$ is a unitary on $SE$, and $\Tr_{\overline{S}_{\rm out}}$ denotes the partial trace over all systems other than $S_{\rm out}$.
Then, consider the following modified protocol.

\begin{itemize}
    \item Step~1': Prepare an incoherent state $\ket{0}_E$ in an ancillary system $E$. 
    Taking the composite system $SE$ as $S'$ in Step~1 above, we run the protocol in Step~1 to prepare states that have coherence for the energy levels of $SE$.
    Here, unlike Step~1, we do not discard the initial state $\rho$ but keep it aside for later use. 
    
    \item Step~2': Run the protocol in Step~2 to amplify the coherence created in Step~1'. 
    \item Step~3': Run the protocol in Step~3 to approximately apply the unitary $V$ --- which appears in the dilution form of the desired channel $\mE$ --- to $\rho\otimes\dm{0}_E$.
    Take the partial trace (which is a covariant operation) over all systems except the output system $S_{\rm out}$. 
\end{itemize}
The final state approximates $\mE(\rho)$, and the error can be made arbitrarily small by making a sufficiently large number of repetitions in the coherence amplification in Step~2' and by preparing a sufficiently large number of copies of the coherent states at the beginning of Step~3'.  

We note that the last step in Step~3 for completely eliminating the correlation between the main system and catalysts (described in Fig.~\ref{fig:remove_correlation}) does not carry over to this case. 
Nevertheless, such correlation can be made arbitrarily small (in the sense that the total state can be made arbitrarily close to a product state) by repeating the coherence amplification procedure in Step~2' many times until the resulting coherent states become sufficiently close to pure states.   
\end{remark}

\begin{remark}
Let us briefly comment on the size of the catalysts required for our protocol. 
Most catalysts are consumed in the coherence amplification in Step~2, which is repeated many times to ensure a sufficient number of coherent states to implement a unitary $V$ in Step~3. 
For an arbitrary target error $\epsilon$, there exist a number $K(\epsilon)$ of running the two-level coherence amplification process in Step~2 and a number $L(\epsilon)$ of qubit coherent states required to implement a unitary $V$ in Step~3, which together realize the final accuracy $\epsilon$.
Since our construction creates $L(\epsilon)$ copies of qubit coherent states (each of which requires $K(\epsilon)$ single-qubit catalysts) for each energy level of the system $S'$ with dimension $d_{S'}$, the dimension of the catalysts used in the protocol is estimated as $\sim 2^{d_{S'} K(\epsilon) L(\epsilon)}$.

We stress that --- similarly to many other works in catalytic resource transformation --- the main objective of Theorem~\ref{thm:transform catalyst} is to present a general protocol that accomplishes powerful transformation tasks and is \emph{not} to find an optimal protocol regarding the catalyst size; indeed, our protocol might be inefficient in the size of catalysts at the cost of its generality. 
A close investigation of the optimal size of the catalysts required to ensure a given target accuracy will certainly make an interesting research program, which we leave for future work. 

\end{remark}

%%%%%%%%%%%%%%%%%%%%%%%%%%%%%%%%%%%%%%%%%%%%%%%%%%%%%%%%%%%%%%%%%%%%%%%%%%%%%%%%%%%%%%%%%%%%%%%%%%%%%%%%%%%%%%

\section{Proof of Proposition~\ref{pro:general monotone marginal catalytic}}

{\it Proposition~\ref{pro:general monotone marginal catalytic}}:
For any given $\mF$ and $\mO_\mF$, suppose that a resource measure $\mfR$ satisfies the superadditivity and the tensor-product additivity.
Then, if $\rho$ is transformable to $\rho'$ by a marginal-catalytic or correlated-catalytic free transformation, it holds that $\mfR(\rho)\geq \mfR(\rho')$.

\bigskip

\begin{proof}
Suppose $\rho$ can be transformed to $\rho'$ by a marginal-catalytic free transformation.
Then, there exists a catalyst $\otimes_{i=0}^{K-1}\tau_{C^{(j)}}$ and free operation $\mE\in\mO_\mF$ such that
\bal
 \mE(\rho\otimes\tau_{C^{(0)}}\otimes\dots\otimes\tau_{C^{(K-1)}}) &= \rho'\otimes \tau_{C^{(0)}\dots C^{(K-1)}}\\
 \Tr_{\overline{C^{(j)}}}\tau_{C^{(0)}\dots C^{(K-1)}}&=\tau_{C^{(j)}},\ \forall j.
\eal

Then, we get
\bal
 \mfR(\rho)+\sum_{i=0}^{K-1}\mfR(\tau_{C^{(j)}})&=\mfR(\rho\otimes\tau_{C^{(0)}}\otimes\dots\otimes\tau_{C^{(K-1)}})\\
 &\geq \mfR(\mE(\rho\otimes\tau_{C^{(0)}}\otimes\dots\otimes\tau_{C^{(K-1)}}))\\
 &= \mfR(\rho'\otimes \tau_{C^{(0)}\dots C^{(K-1)}}')\\ 
 &\geq \mfR(\rho')+\sum_{i=0}^{K-1}\mfR(\tau_{C^{(j)}}),
\eal
where we used the tensor-product additivity in the first line, monotonicity under free operations in the second line, and the tensor-product additivity and the superadditivity in the fourth line.   
The case for the correlated catalysts can be shown similarly. 

\end{proof}

%%%%%%%%%%%%%%%%%%%%%%%%%%%%%%%%%%%%%%%%%%%%%%%%%%%%%%%%%%%%%%%%%%%%%%%%%%%%%%%%%%%%%%%%%%%%%%%%%%%%%%%%%%%%%%

\section{Proof of Proposition~\ref{pro:sufficient}}

{\it Proposition~\ref{pro:sufficient}}:
For any given $\mF$ and $\mO_\mF$, suppose that $\mO_\mF$ includes the relabeling of the classical register and free operations conditioned on the classical register.
Then, if $\rho$ is asymptotically transformable to $\rho'$, there exists a free transformation from $\rho$ to $\rho'$ with a correlated catalyst as well as marginal catalysts with an arbitrarily small error. 

\bigskip

\begin{proof}

A proof for the correlated catalysts essentially follows the technique introduced in Ref.~\cite{Shiraishi2021quantum}. (See also Fig.~\ref{fig:swap}.)
We say that $\rho$ can be asymptotically transformed to $\rho'$ if for any $\epsilon>0$ there exists an integer $n$ and a free operation $\mE\in\mO_\mF$ such that 
\bal
 \frac{1}{2}\|\mE(\rho^{\otimes n})-{\rho'}^{\otimes n}\|_1\leq \epsilon.
 \label{eq:error asymptotic}
\eal
Although the spaces that $\rho$ and $\rho'$ act on can be different, we can always choose a larger space $S$ that contains the supports of both states.  
Thus, we can set $\rho,\rho'\in\mD(S)$ without the loss of generality, and in particular, $\rho^{\otimes n},{\rho'}^{\otimes n}\in\mD(\otimes_{i=1}^{n} S_i)$ using the isomorphic subspaces $S_i\cong S,\forall i$. 
Let $\Xi\coloneqq\mE(\rho^{\otimes n})$ be a state satisfying \eqref{eq:error asymptotic} and $\Xi_i\coloneqq\Tr_{i+1\dots n}\Xi$ be the marginal state of $\Xi$. 
We define a catalyst $\tau$ as
\bal
 \tau = \frac{1}{n}\sum_{k=1}^{n}\rho^{\otimes k-1}\otimes \Xi_{n-k}\otimes\dm{k}_R
\eal
where $\rho^{\otimes k-1}\in\mD(S_2\otimes\dots\otimes S_{k})$, $\Xi_{n-k}\in\mD(S_{k+1}\otimes\dots\otimes S_{n})$, and $\{\ket{k}_R\}$ is a set of orthonormal states defined in the classical register system $R$.
We consider applying a free operation to an initial state of the form 
\bal
 \rho\otimes\tau = \frac{1}{n}\sum_{k=1}^{n}\rho^{\otimes k}\otimes \Xi_{n-k}\otimes\dm{k}_R,
\eal
which acts on $S_1\otimes\dots\otimes S_{n}\otimes R$.
Let $\mM_k(\cdot)=\dm{k}\cdot\dm{k}$, $\mM_{\overline k}(\cdot)=(\mbI-\dm{k})\cdot(\mbI-\dm{k})$ be projectors acting on $R$ and define a channel $\mE_{\rm cond}$ as 
\bal
 \mE_{\rm cond}=\mE\otimes \mM_{n} + \id_{S_1\dots S_{n}}\otimes \mM_{\overline{n}}.
 \label{eq:conditional free}
\eal
This corresponds to a conditional application of the free operation $\mE$ when the classical register has $k=n$ but does nothing otherwise.
By assumption, $\mE_{\rm cond}$ is a free operation.
Let us also define $\mR(\cdot)$ to be the operation that cyclically shifts the classical labels as $1\to 2, 2\to 3,\dots, n\to 1$.
We apply this relabeling locally to the classical register by 
\bal
 \mE_{\rm relab}=\id_{S_1\dots S_{n}}\otimes \mR,
 \label{eq:relabeling}
\eal
which is also a free operation by assumption. 
Finally, we cyclically swap the systems by $\mE_{\rm swap}$ that shifts $S_i\to S_{i+1}$ with $S_{n+1}\coloneqq S_1$.
Since $S_{1},\dots,S_{n}$ are all isomorphic to each other, swapping these systems are also free~\cite{Chitambar2019quantum}.

The concatenation of these channels $\Lambda\coloneqq\mE_{\rm swap}\circ\mE_{\rm relab}\circ\mE_{\rm cond}$ satisfies that 
\bal
 \Tr_{\overline{S}_1}\Lambda(\rho\otimes\tau)=\frac{1}{n}\sum_{k=1}^{n}\Tr_{\overline{k}}\Xi
\eal
while
\bal
 \Tr_{S_1}\Lambda(\rho\otimes\tau)=\tau.
\eal
Since 
\bal
 \frac{1}{2}\|\Tr_{\overline{S}_1}\Lambda(\rho\otimes\tau)-\rho'\|_1&= \frac{1}{2}\left\|\frac{1}{n}\sum_{k=1}^{n}\Tr_{\overline{k}}\Xi-\rho'\right\|_1\\
 &\leq \frac{1}{2n}\sum_{k=1}^{n}\left\|\Tr_{\overline{k}}\Xi-\rho'\right\|_1\\
 &\leq \frac{1}{2n}\sum_{k=1}^{n}\left\|\Xi-{\rho'}^{\otimes n}\right\|_1\\
 &\leq \epsilon,
\eal
the free operation $\Lambda$ allows for a desired catalytic transformation with correlated catalyst $\tau$ with accuracy $\epsilon$.

Finally, the above correlated-catalytic free transformation can be turned into a marginal catalytic free transformation by the same technique used at the end of the proof for Theorem~\ref{thm:transform catalyst}, namely adding another catalytic state that contains the final state and swapping this into the system. 

\begin{figure}
    \centering
    \includegraphics[width=0.6\textwidth]{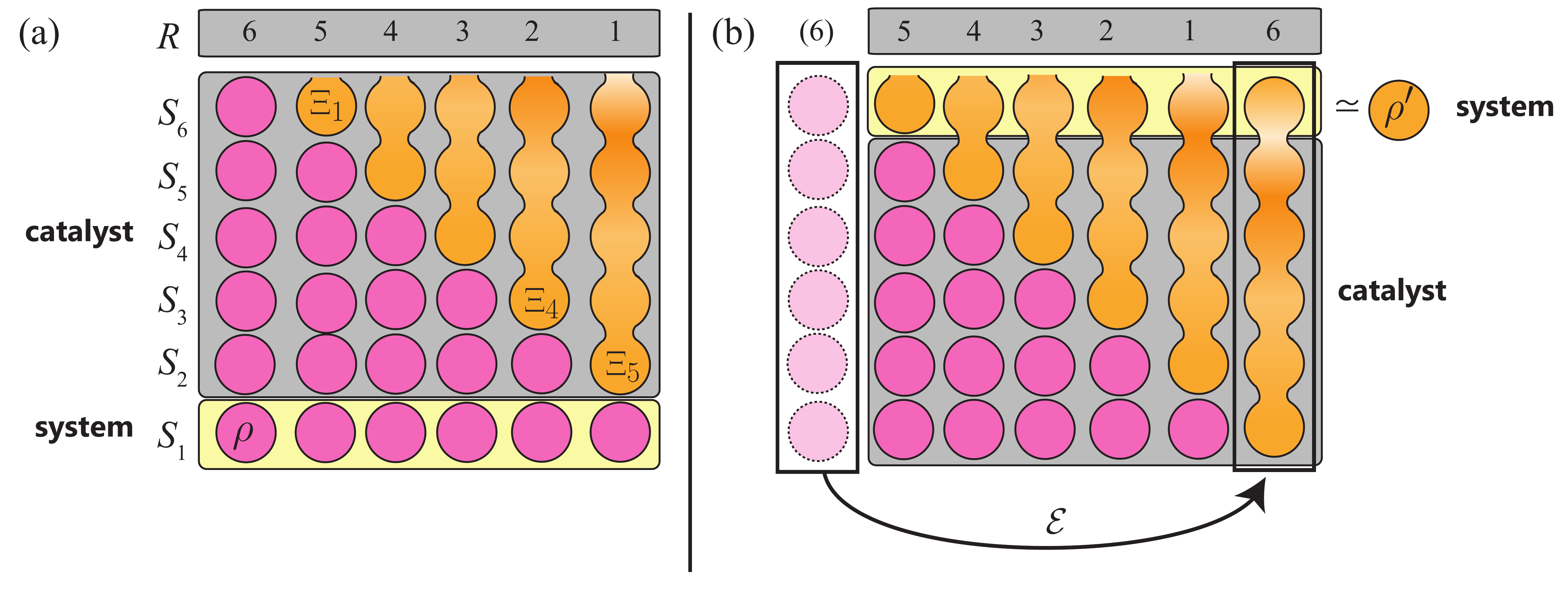}
    \caption{Schematic of how asymptotic transformation reduces to correlated-catalytic transformation.
    Note that the register system $R$ is included in the catalyst.
    We apply the asymptotic transformation on the composite system of the system ($S_1$) and the catalyst ($S_2\sim S_6$) if the state of the register system is $\ket{6}$.
    }
    \label{fig:swap}
\end{figure}

\end{proof}

\begin{remark}
We remark that the assumption on the ability to relabel the classical register and to apply free operations conditioned on the classical register is a very mild one and is satisfied by essentially all important resource theories such as entanglement, (speakable and unspeakable) coherence, and thermal non-equilibrium.
\end{remark}

\section{Correlated-catalytic coherence transformation}

Here, we formally state our conjecture and present some relevant observations.
Firstly, it is important to notice that exactly the same power we get for the marginal catalysts cannot be expected for correlated catalysts because the coherence no-broadcasting theorem prohibits one from creating non-zero coherence from incoherent states by correlated-catalytic covariant transformations, whereas marginal-catalytic transformations allow this as we showed in Theorem~\ref{thm:transform catalyst}.  
The question is whether we can realize arbitrary transformations with the help of correlated catalysts if the initial states have non-zero coherence. 

As introduced in the main text, we define the set of index pairs specifying the energy differences for which $\rho$ has non-zero coherence as  
\bal
\mI(\rho):=\lset (i,j) \sbar {}_S\bra{i}\rho\ket{j}_S\neq 0\rset
\eal
and also introduce the set of index pairs for the energy differences in the target system $S'$ that match a linear combination of integer multiples of the energy differences specified by $\mI(\rho)$ as 
\bal
 \mJ(\rho):=\lset(i',j')\sbar E_{S',j'}-E_{S',i'}=\sum_{(k,l)\in\mI(\rho)} m_{i'j'}^{kl} (E_{S,l}-E_{S,k}),\  m_{i'j'}^{kl}\in\mbZ\rset.
\label{eq:label pair whole}
\eal
We recast the conjecture.

\begin{quote} %\label{conj:transform catalyst}
{\bf Conjecture \ref{conj:transform catalyst}.}
For any $\rho\in\mD(S)$, $\rho'\in\mD(S')$ and $\epsilon>0$, $\rho$ can be transformed to a state $\rho'_\epsilon\in\mD(S')$ such that $\frac{1}{2}\|\rho'-\rho'_\epsilon\|_1\leq \epsilon$ by a correlated-catalytic covariant transformation if and only if $\mI(\rho')\subseteq\mJ(\rho)$.
\end{quote}

Below, we present three observations on why this conjecture can be expected. 

\bigskip

\textit{\textbf{Quasi-correlated-catalytic transformations.}}
---
In general, marginal-catalytic transformations may generate correlation between different catalysts at the end of the protocol. 
However, by taking a close look at our protocol for Theorem~\ref{thm:transform catalyst}, one may notice that only the operations in Step~1 actively create correlation in the catalytic systems.
The other operations only create correlation with the main system or the ancillary system outside the catalytic systems, and the correlation in the catalytic systems arises when the other systems are traced out; this is intuitively a `passive' correlation generated through the `active' correlations between the catalytic systems and the external systems. 
We formalize this idea by considering a class of transformations realized by sequences of correlated-catalytic transformations, which we call \emph{quasi-correlated-catalytic} covariant transformations.

\begin{defn}

Suppose $\mE:\mD(SC^{(0)}\dots C^{(K-1)})\rightarrow \mD(S'C^{(0)}\dots C^{(K-1)})$ is a covariant operation realized as $\mE=\Lambda_{K-1}\circ\dots\circ\Lambda_0$ where each $\Lambda_i:X_iC^{(i)}\to Y_iC^{(i)}$ is a covariant operation constructing a correlated-catalytic covariant transformation from $X_i$ to $Y_i$, both of which are non-catalytic systems, with the catalytic system $C^{(i)}$.
Then, if such $\mE$ can bring $\rho\in\mD(S)$ to $\rho'\in\mD(S')$ as $\Tr_{\overline{S'}}\mE(\rho\otimes\tau_{C^{(0)}}\dots\otimes\tau_{C^{(K-1)}})= \rho'$, we say that $\rho$ can be transformed to $\rho'$ by a quasi-correlated-catalytic covariant transformation.

\end{defn}

Note that if we employ an ancillary system whose initial state is incoherent (such as the ancillary system $R$ used for the protocol in Theorem~\ref{thm:transform catalyst}), it can also be taken as a subsystem of $X_i$ and $Y_i$ above; as long as $\Lambda_i$ keeps the state in $C^{(i)}$ intact, $\Lambda_i$ can serve as a constituent of a quasi-correlated-catalytic transformation. 

The following result shows that the `if' part of the conjecture holds for this extended class of correlated-catalytic transformations. 

\begin{pro}
For any $\rho\in\mD(S)$, $\rho'\in\mD(S')$ and $\epsilon>0$, $\rho$ can be transformed to a state $\rho'_\epsilon\in\mD(S')$ such that $\frac{1}{2}\|\rho'-\rho'_\epsilon\|_1\leq \epsilon$ by a quasi-correlated-catalytic covariant transformation if $\mI(\rho')\subseteq\mJ(\rho)$.
\end{pro}
\begin{proof}

The proof strategy is similar to the one for Theorem~\ref{thm:transform catalyst}. 
However, since Step~1 in the proof of Theorem~\ref{thm:transform catalyst} uses operations that  actively create correlation inside the catalytic subsystems, we need to employ a different strategy to create small coherence in the ancillary system. 
Here, we instead extract a small coherence from the initial state $\rho$, which has non-zero coherence for the energy differences specified by the index pairs in $\mI(\rho)$.
We will show that one can use this `seed' coherence to cover coherences for other energy differences that are integer multiples of the seed and linear combinations thereof. 
To make the comparison to Theorem~\ref{thm:transform catalyst} explicit, we break down our protocol into three steps.

\textit{Step~1: Creating small coherence.}\quad 
For every $(i,j)\in\mI(\rho)$, we prepare $\ket{0}_{R_{ij}}$, the energy eigenstate acting on a two-level system $R_{ij}$ with Hamiltonian
\bal
 H_{R_{ij}} = (E_{S,j}-E_{S,i}) \dm{1}_{R_{ij}}. 
\eal
Let $U_{SR_{ij}}$ be a unitary matrix acting on $SR_{ij}$ defined as 
\bal
U_{S R_{ij}}:=\dm{i0}+\dm{j1}+\frac{1}{\sqrt{2}}(\ketbra{i1}{i1}+\ketbra{i1}{j0}+\ketbra{j0}{i1}-\ketbra{j0}{j0})+\mbI_{\bar i \bar j}\otimes \mbI_{R_{ij}}
\label{eq:unitary to create seed}
\eal
where $\mbI_{\bar i \bar j}=\sum_{k\neq i,j}\dm{k}$ is the identity matrix acting on the subspace that does not intersect ${\rm span}\{\ket{i},\ket{j}\}$. 
Since this preserves the total energy, the application of this unitary is covariant. 
However, because of the nontrivial transformation within the degenerate subspace ${\rm span}\{\ket{i1}, \ket{j0}\}$, non-zero coherence can be transferred from $\rho_{S}$ to $\ket{0}_{R_{ij}}$. 
Then, by sequentially applying this unitary over the state in $S$ and other $R_{ij}$'s, we can realize non-zero coherence on every $R_{ij}$.
Below, we make this intuition precise.

For the explanatory purpose, we arrange the indices in $\mI(\rho)$ in some order (e.g., lexicographic order) and if $(i,j)$ is the $k$\,th element in $\mI(\rho)$, we call the system associated with this pair of indices as $R^{(k)}:=R_{ij}$. (Here, we let $k$ begin with $k=1$, i.e., $k=1,\dots,|\mI(\rho)|$ where $|\mI(\rho)|$ is the number of elements in $\mI(\rho)$.)
Letting $\rho_{S \{R^{(k)}\}}^{(0)}:=\rho_{S}\otimes\dm{0}_{R^{(1)}}\otimes\dots\otimes\dm{0}_{R^{(|\mI(\rho)|)}}$, we sequentially apply the unitary in \eqref{eq:unitary to create seed} to obtain
\bal
 \rho_{S \{R^{(k)}\}}^{(l+1)} = (\mbI_{\bar R^{(l+1)}} \otimes U_{S R^{(l+1)}}) \rho_{S \{R^{(k)}\}}^{(l)} (\mbI_{\bar R^{(l+1)}} \otimes U_{S R^{(l+1)}}^\dagger)
 \label{eq:asymmetry distribution}
\eal
for $l=0,1,\dots,|\mI(\rho)|-1$, where $\mbI_{\bar R^{(l+1)}}$ is the identity matrix acting on the reference systems $R^{(1)}\dots R^{(|\mI(\rho)|)}$ except $R^{(l+1)}$.
From now on, we simply write the label for the relevant system in the subscript to refer to the reduced state acting on the specified system, e.g., $\rho_{R^{(l)}}^{(l)}$ refers to the reduced state obtained by taking the partial trace of $\rho_{S \{R^{(k)}\}}^{(l)}$ over systems other than $R^{(l)}$.

The following lemma states that this protocol creates non-zero coherence in every reference system.

\begin{lem}
 For any $l=1,\dots,|\mI(\rho)|$, it holds that $\bra{0}\rho_{R^{(l)}}^{(l)}\ket{1}\neq 0$. 
\end{lem}
\begin{proof}
Let $(i,j)\in\mI(\rho)$ be the $l$+1\,th element of $\mI(\rho)$, i.e., $R^{(l+1)}=R_{ij}$. 
First, note that 
\bal
 \bra{0}\rho_{R_{ij}}^{(l+1)}\ket{1}&=\bra{0}\Tr_{S}\left[ U_{S R^{(l+1)}}\left( \rho_{S}^{(l)}\otimes \dm{0}_{R^{(l+1)}}\right) U_{S R^{(l+1)}}^\dagger \right] \ket{1}\\
 &=\frac{1}{\sqrt{2}}\bra{i}\rho_{S^{(l)}}\ket{j}.
\eal
This implies that if $\bra{i}\rho_{S}^{(l)}\ket{j}\neq 0$, it always holds that $\bra{0}\rho_{R_{ij}}^{(l+1)}\ket{1}\neq0$. 
Thus, it suffices to show that
\bal 
\mI(\rho)=\mI(\rho_{S}^{(0)})\subseteq \mI(\rho_{S}^{(l)})
\eal
for any $l$ to prove the statement.

For this, we prove that the application of unitary in \eqref{eq:unitary to create seed} always maps non-zero off-diagonal elements in $S$ to non-zero off-diagonal elements in $S$, implying that $\mI(\rho_{S}^{(l)})\subseteq \mI(\rho_{S}^{(l+1)})$.
Since only $i,j$\,th rows and columns of $\rho_{S}^{(l)}$ can get affected by the unitary in \eqref{eq:unitary to create seed}, it suffices to show that if all of the off-diagonal elements of $\rho_S^{(l)}$ with $i$, $j$, and $k$-th rows and columns are non-zero, then those of $\rho_S^{(l+1)}$ are also non-zero.
This can be directly checked as  
\bal
 \bra{i}\rho_{S}^{(l+1)}\ket{k} &= \bra{i}\Tr_{R_{ij}}\left[ U_{S R_{ij}}\left( \rho_{S}^{(l)}\otimes \dm{0}_{R_{ij}}\right) U_{S R_{ij}}^\dagger\right]\ket{k}\\
 &= \bra{i}\rho_{S}^{(l)}\ket{k}
\eal
and
\bal
 \bra{j}\rho_{S}^{(l+1)}\ket{k} &= \bra{j}\Tr_{R_{ij}}\left[ U_{S R_{ij}}\left( \rho_{S}^{(l)}\otimes \dm{0}_{R_{ij}}\right) U_{S R_{ij}}^\dagger\right]\ket{k}\\
 &= -\frac{1}{\sqrt{2}}\bra{j}\rho_{S}^{(l)}\ket{k}
\eal
for $k\neq i,j$, and
\bal
 \bra{i}\rho_{S}^{(l+1)}\ket{j} &= \bra{i}\Tr_{R_{ij}}\left[ U_{S R_{ij}}\left( \rho_{S}^{(l)}\otimes \dm{0}_{R_{ij}}\right) U_{S R_{ij}}^\dagger\right]\ket{j}\\
 &= -\frac{1}{\sqrt{2}}\bra{i}\rho_{S}^{(l)}\ket{j},
\eal
which concludes the proof.
\end{proof}

\textit{Step~2: Amplifying coherence.}\quad
The amplification of the coherence in $R_{ij}$ for $(i,j)\in\mI(\rho)$ is accomplished in the same way as Step~2 in the protocol for Theorem~\ref{thm:transform catalyst}.
Importantly, since every step in the two-level coherence amplification is applied over a catalytic subsystem and the ancillary system, each step in Step~2 is done by an operation that does not actively create coherence in the catalytic subsystems. 

\textit{Step~3: Prepare the target state.}\quad
Similarly to Step~3 for Theorem~\ref{thm:transform catalyst}, we repeat Step~1 and 2 for $L(\gg 1)$ times to prepare $\ket{+}_{R_{ij}}^{\otimes L}$ for each $(i,j)\in\mI(\rho)$ and use it as an ancillary coherent resource state to implement a desired unitary on the main system to prepare the target state. 
However, since we only have coherence for the energy differences corresponding to $\mI(\rho)$ unlike the case of Theorem~\ref{thm:transform catalyst}, we need to carefully consider how we prepare a state close to $\rho'$, which has coherence for energy differences in $\mJ(\rho)$. 
Intuitively, combining the multiple applications of the energy-shifting operator \eqref{eq:energy-shifting} on different ancillary subsystems holding coherences for different energy differences should allow us to realize the coherences in $\mJ(\rho)$.
Below, we make this intuition precise.

Let $\mS\subseteq\{0,\dots,d_{S'}-1\}$ be a subset of integers. 
We call $\mS$ \emph{closed index set} if $(i,j)\in\mJ(\rho),\forall i,j\in\mS$.
Since there exist many closed index sets in general, it is convenient to consider the largest sets of this type; we specifically call $\mS$ \emph{maximal closed index set} if for any $i\not\in\mS$, there exists $j\in\mS$ such that $(i,j)\not\in\mJ(\rho)$.
We first show that the set of maximal closed index sets completely partitions $\{0,\dots,d_{S'}-1\}$.

\begin{lem} \label{lem:maximal index sets}
For any $i\in\{0,\dots,d_{S'}-1\}$, $i$ belongs to one and only one maximal closed index set.  
\end{lem}
\begin{proof}
Firstly, $i\in\{0,\dots,d_{S'}-1\}$ belongs to at least one maximal closed index set because $\{i\}$ is always a closed index set, and any closed index set is a subset of some maximal closed index set. 
Thus, it suffices to show that $i$ only belongs to one maximal index set. 
Let us assume, to the contrary, that $i$ belongs to two distinct maximal closed index sets $\mS_0$ and $\mS_1$.
Since $\mS_0\not\subset\mS_1$ and $\mS_1\not\subset\mS_0$ by the assumption that $\mS_0$ and $\mS_1$ are maximal closed index sets, each set has at least two elements. 
% there exists two other integers $j,k$ such that $j\in\mS_0\setminus\mS_1$ and $k\in\mS_1\setminus\mS_0$. 
Pick two distinct integers $j\in\mS_0\setminus\mS_1$ and $k\in\mS_1$. 
Then, there exist sets of integers $\{m_{ij}^{ln}\}_{ln}$ and $\{m_{ik}^{ln}\}_{ln}$ such that 
\bal
 E_{S',j}-E_{S',i}&=\sum_{(l,n)\in\mI(\rho)} m_{ij}^{ln}(E_{S,n}-E_{S,l})\\
 E_{S',k}-E_{S',i}&=\sum_{(l,n)\in\mI(\rho)} m_{ik}^{ln}(E_{S,n}-E_{S,l}).
\eal
However, since $E_{S',k}-E_{S',j}=(E_{S',k}-E_{S',i})-(E_{S',j}-E_{S',i})$, we also get 
\bal
 E_{S',k}-E_{S',j}=\sum_{(l,n)\in\mI(\rho)} \left(m_{ik}^{ln}-m_{ij}^{ln}\right)(E_{S,n}-E_{S,l}),
\eal
implying that $(j,k)\in\mJ(\rho)$.
Since this holds for any $k\in\mS_1$ for fixed $j$, we get that $j\in\mS_1$, which is a contradiction. 
\end{proof}

Lem.~\ref{lem:maximal index sets} implies that maximal index sets $\mS_0,\mS_1,\dots\mS_T$ with $T\leq d_{S'}-1$ completely partition $\{0,\dots,d_{S'}-1\}$.
As a result, the Hilbert space $\mH'$ underlying the target system $S'$ can be decomposed as $\mH'=\bigoplus_{l}\mH_l$ where $\mH_l:={\rm span}\lset\ket{i}\sbar i\in\mS_l\rset$. 
Then, if $\mI(\rho')\subseteq\mJ(\rho)$, the definition of closed label sets allows $\rho'$ to admit a block diagonal form 
\bal
 \rho'= \bigoplus_l p_l \rho_l',\ \ \supp(\rho_l')\subseteq\mH_l.
\eal

It now suffices to show that each $\rho_l'$ can be obtained with an arbitrary accuracy by covariant operation applied over $S'$ and $R$, since $\rho'$ can be realized by stochastically prepare $\rho_l'$ at probability $p_l$, and any operation that solely acts on the systems outside the catalytic subsystems clearly does not actively create correlation in the catalytic subsystems.
(Note that, as in the case for Theorem~\ref{thm:transform catalyst}, we discard the ancillary system $R$ at the end of the protocol, and it is not a catalytic subsystem.)
Moreover, it suffices to show that any pure state acting on $\mH_l$ can be obtained, because then any mixed state $\rho_l'$ can be obtained by stochastically preparing the pure states that appear in a convex combination that constitutes $\rho_l'$.
Thus, we focus on the case where $\rho_l$ is pure for all $l$. 
In this case, for any $l$ and another arbitrary pure state $\ket{\psi_l}\in\mH_l$, there exists a unitary $V_l$ acting on $\mH_l$ such that
\bal
\rho_l'=V_l\dm{\psi_l}_{S'}V_l^\dagger.
\label{eq:goal pure}
\eal 
Specifically, we choose $\ket{\psi_l}=\ket{j_l^\star}$ where $\ket{j_l^\star}$ is an energy eigenstate of $H_{S'}$ and $j_l^\star$ is an arbitrary integer such that $j_l^\star\in\mS_l$.
For any $k\in\mS_l$, let $\left\{m_{kj_l^\star}^{ij}\right\}_{ij}$ be a set of integers such that 
\bal
E_{S',j_l^\star}-E_{S',k}=\sum_{ij}m_{kj_l^\star}^{ij}(E_{S,j}-E_{S,i}),
\label{eq:energy resonance subspace}
\eal
whose existence is ensured by the definition of $\mS_l$ and \eqref{eq:label pair whole}. 

Similarly to the procedure for Theorem~\ref{thm:transform catalyst}, we embed the state on $R_{ij}^{\otimes L}$ into another system $\tilde R_{ij}^{L}$ with Hamiltonian 
\bal
H_{\tilde R_{ij}^{L}}=H_{R_{ij}^{\otimes L}}+(E_{S',j}-E_{S',i})\left(\sum_{t=-M_{ij}}^{-1}t\dm{t}+\sum_{t=L+1}^{L+M_{ij}}t\dm{t}\right),
\eal
where $M_{ij}\coloneqq\max_{k\in\mS_l} m_{kj_l^\star}^{ij}$ and $\ket{t}$ with $t\not\in\{0,\dots,L\}$ is an eigenstate of the new Hamiltonian orthogonal to all the eigenstates of $H_{R_i^{\otimes L}}$. 
Let $\Delta_{A}$ be the energy-shifting operator acting on system $A$ defined as 
$\Delta_{A}(j) = \sum_{k=0}^{L} \ketbra{k-j}{k}$
and $P_{\tilde R_{ij}^L}^{0:L}\coloneqq\sum_{n=0}^L\dm{n}_{\tilde R_{ij}^L}$ where $\ket{n}$ is the eigenstate of $H_{\tilde R_{ij}^L}$ as in \eqref{eq:resource binomial}.
Then, we consider a covariant operation $\mE_{V_l}:\mD(\mH_l\otimes\tilde R)\rightarrow\mD(\mH_l\otimes\tilde R)$ with Kraus operators

\bal
 K_0 := \sum_{k\in\mS_l}\dm{k}V_l\dm{j_l^\star}\bigotimes_{(i,j)\in\mI(\rho)}\Delta_{\tilde R_{ij}}\left(m_{kj_l^\star}^{ij}\right) 
\eal
\bal
 K_1 := (\mbI_{\mH_l}-\dm{j_l^\star})\bigotimes_{(i,j)\in\mI(\rho)}P_{\tilde R_{ij}^L}^{0:L}+\mbI_{\mH_l}\otimes\left[\mbI-\bigotimes_{(i,j)\in\mI(\rho)}P_{\tilde R_{ij}^L}^{0:L}\right]
\eal

If we apply the above channel to $\dm{j_l^\star}\bigotimes_{(i,j)\in\mI(\rho)}\dm{+}_{\tilde R_{ij}^L}^{\otimes L}$ and take the partial trace over $\tilde R$ systems, $K_0$ deterministically clicks and gives  
\bal
\Tr_{\tilde R}&\left[K_0 \left( \dm{j_l^\star}\bigotimes_{(i,j)\in\mI(\rho)}\dm{+}_{\tilde R_{ij}^L}^{\otimes L} \right) K_0^\dagger\right] \\&= \sum_{k\in\mS_l} \sum_{p\neq k} \dm{k}V\dm{j_l^\star}V^\dagger\dm{p} \prod_{(i,j)\in\mI(\rho)}\Tr\left[\Delta_{\tilde R_{ij}^L}\left(m_{kj_l^\star}^{ij}-m_{pj_l^\star}^{ij}\right)\dm{+}_{\tilde R_{ij}^L}^{\otimes L}\right]+\sum_{k\in\mS_l} \dm{k}V\dm{j_l^\star}V^\dagger\dm{k}
\eal

Analogously to Theorem~\ref{thm:transform catalyst}, it suffices to show $\Tr\left[\Delta_{\tilde R_{ij}^L}\left(m\right)\dm{+}_{\tilde R_{ij}^L}^{\otimes L}\right]\to 1$ as $L\to\infty$ for every $(i,j)\in\mI(\rho)$ and $|m|\leq 2 M_{ij}$. 
We have 
\bal
 \Tr\left[\Delta_{\tilde R_{ij}^L}(m)\dm{+}_{\tilde R_{ij}^L}^{\otimes L}\right]&=\frac{1}{2^L}\sum_{n=0}^{L-|m|}\binom{L}{n+|m|}^{1/2}\binom{L}{n}^{1/2}.
 \label{eq:shift overlap quasi}
\eal
It holds that 
\bal
\begin{cases}
 \binom{L}{n+|m|}\geq \binom{L}{n} & n\leq  \frac{L}{2}-\frac{|m|}{2} \\
 \binom{L}{n}\geq \binom{L}{n+|m|} & n\geq  \frac{L}{2} -\frac{|m|}{2}
\end{cases}.
\eal
When $L$ is even, \eqref{eq:shift overlap quasi} can be bounded as 
\bal
 \Tr\left[\Delta_{\tilde R_{ij}^L}(m)\dm{+}_{\tilde R_k^L}^{\otimes L}\right]
 &=\frac{1}{2^L}\sum_{n=0}^{L-|m|} \binom{L}{n}^{1/2}\binom{L}{n+|m|}^{1/2} \\
 &\geq \frac{1}{2^L}\sum_{n=0}^{\lfloor\frac{L}{2}-\frac{|m|}{2}\rfloor}\binom{L}{n}+\frac{1}{2^L}\sum_{n=\lfloor\frac{L}{2}-\frac{|m|}{2}\rfloor+1}^{L-|m|}\binom{L}{n+|m|}\\
 &= \frac{1}{2^L}\sum_{n=0}^{\lfloor\frac{L}{2}-\frac{|m|}{2}\rfloor}\binom{L}{n}+\frac{1}{2^L}\sum_{n=\lfloor\frac{L}{2}+\frac{|m|}{2}\rfloor+1}^{L}\binom{L}{n}\\
 &=1-\frac{1}{2^L}\sum_{n=\lfloor\frac{L}{2}-\frac{|m|}{2}\rfloor+1}^{\lfloor\frac{L}{2}+\frac{|m|}{2}\rfloor}\binom{L}{n}\\
 &\geq 1-\frac{|m|}{2^L}\binom{L}{L/2}\\
 &\geq 1-\frac{2|m|}{\sqrt{2\pi}}L^{-1/2}-\delta_e\xrightarrow[L\to\infty]{} 1-\delta_e
\eal
for any $\delta_e>0$.

When $L$ is odd, we can follow the same argument to get 
\bal
 \Tr\left[\Delta_{\tilde R_{ij}^L}(m)\dm{+}_{\tilde R_k^L}^{\otimes L}\right]&\geq 1-\frac{|m|}{2^L}\binom{L}{\frac{L+1}{2}}\\
 &\geq 1-\frac{2|m|}{\sqrt{2\pi}}(1+1/L)^{-L/2}(1-1/L)^{-L/2}\left(\sqrt{L}+1/\sqrt{L}\right)^{-1}-\delta_o\xrightarrow[L\to\infty]{} 1-\delta_o
\eal
for any $\delta_o>0$.
These imply $\lim_{L\to\infty}\Tr\left[\Delta_{\tilde R_{ij}^L}\left(m\right)\dm{+}_{\tilde R_{ij}^L}^{\otimes L}\right]= 1$ for every $(i,j)\in\mI(\rho)$ and $|m|\leq 2 M_{ij}$, concluding the proof.

\end{proof}

Besides the state transformability under the quasi-correlated catalytic transformations, we have two other observations relevant to the conjecture.

\bigskip

\textit{\textbf{Qubit transformations.}}
---
We considered a slightly extended class of correlated-catalytic covariant transformations. 
What if we stick to the original one? 
In Ref.~\cite{Ding2021amplifying}, the authors numerically studied the sets of states achievable from a weakly coherent state using correlated-catalytic covariant operations with low-dimensional catalysts.   
They investigated up to three-dimensional catalysts and confirmed that higher-dimensional catalysts admit larger sets of achievable states, including states with larger coherence.  
If this tendency persists to higher catalyst dimensions, it can be expected that any two-dimensional state can be reached from any coherent initial state with arbitrary precision by employing a catalyst of sufficiently large dimensions.

\bigskip

\textit{\textbf{Creating non-zero coherence in higher modes.}}
---
The above observation only applies to qubit systems. 
When higher-dimensional systems are in question, the structure of coherence becomes more involved, as coherences for different energy differences are not comparable in general. 
Nevertheless, we can show that in any three-level system with eigenstates 1, 2, and 3, if the initial state has coherence between 1 and 2, 2 and 3, but not 1 and 3, we can {\it broadcast} the coherence to the level spacing between 1 and 3 with the help of correlated catalyst. 

\begin{pro}
Consider a three-level system with energies $E_1$, $E_2$ and $E_3$ with corresponding energy eigenstates $\ket{E_1}$, $\ket{E_2}$, and $\ket{E_3}$.
Then, for any $3\times 3$ density matrix $\rho$ with $\rho_{13}=\rho_{31}=0$ and other non-zero matrix elements:
\bal
\rho=\begin{pmatrix}
 \rho_{11}& \rho_{12}&0 \\
\rho_{21}&\rho_{22}&\rho_{23} \\
0&\rho_{32}&\rho_{33}
\end{pmatrix},
\eal
there exists a three-state catalytic system $C$ with the same energy as $S$, a state of $C$ denoted by $\tau$, and a covariant channel $\mE$ such that $\sigma =\mE(\rho\otimes \tau)$ with $\Tr_S[\sigma]=\tau$ and $\rho':=\Tr_C[\sigma]$ with $\rho'_{13}\neq 0$.
\end{pro}

We stress that the ratios of two energy differences of $\Di E_{12}$, $\Di E_{23}$, and $\Di E_{13}$ (e.g., $\Di E_{12}/\Di E_{23}$) are irrational numbers in general.
Our result shows that the coherence for $\Di E_{13}$ in the catalyst $\tau$ can indeed be broadcast to $\rho'$ under the presence of non-zero coherence for $\Di E_{12}$ and $\Di E_{13}$ in $\rho$.
Note that this does not contradict the coherence no-broadcasting theorem~\cite{Lostaglio2019coherence,Marvian2019nobroadcasting}, which prohibits coherence from being broadcast to states that do not have coherence for \emph{any} energy difference.

\begin{proof}
Our channel consists of two maps.
First, we consider a covariant channel $\mE_1$ on $S$ that acts as 
\bal
\mE_1(\rho) = \tilde\rho:= \begin{pmatrix}
\frac14 &\delta & 0 \\
\delta & \frac12 & \delta \\
0 & \delta & \frac14
\end{pmatrix}
\eal
with some real positive number $0<\delta<D$, where the upper bound $D$ is introduced for later use.
This map is easily realized by applying phase rotation on $\ket{E_1}$ and $\ket{E_3}$, decoherence, and a classical stochastic map on diagonal elements.

Then, to construct our second map $\mE_2$ we define the maximally entangled states on $SC$ as
\bal
\ket{e^{\pm}_{ij}}\coloneqq\frac{1}{\sqrt{2}}(\ket{E_iE_j}\pm \ket{E_jE_i})
\eal
with $(i,j)=(1,2), (2,3), (1,3)$.
We set the catalyst in $C$ as
\bal
\tau=\begin{pmatrix}
\frac{1}{10}+\frac{152}{225}\delta^2 & \frac{38}{45}\delta  &\frac{152}{225}\delta^2 \\
\frac{38}{45}\delta &\frac45-\frac{304}{225}\delta^2 & \frac{38}{45}\delta \\
\frac{152}{225}\delta^2 & \frac{38}{45}\delta & \frac{1}{10}+\frac{152}{225}\delta^2
\end{pmatrix}.
\label{sigma}
\eal
The upper bound $D$ is set to keep $\tau$ positive semidefinite. (Since $a\to 0$ with $\delta \to 0$, there exists some $D>0$ satisfying it.)
We construct our channel $\mE_2$ on $SC$ as
\bal
\mE_2 (\tilde\rho\otimes \tau):=K_0 (\tilde\rho\otimes \tau) K_0^\dagger + \sum_{k=1}^3 \sum_{(i,j)=(1,2), (2,3), (1,3)} K_{k, (i,j)} (\tilde\rho\otimes \tau) K_{k,(i,j)}^\dagger ,
\eal
where $K_0$ is a projection operator onto the subspace spanned by $\{ \ket{E_iE_i}\}_{i=1}^3\oplus \{ \ket{e^+_{ij}}\} _{(i,j)=(1,2), (2,3), (1,3)}$
\bal
K_0:=\sum_{i=1}^3 \ket{E_iE_i}\bra{E_iE_i} +\sum_{(i,j)=(1,2), (2,3), (1,3)} \ket{e^+_{ij}}\bra{e^+_{ij}},
\eal
and $K_{k, (i,j)}$'s are defined as 
\bal
K_{k, (i,j)} :=p_{k,(i,j)} \ket{E_kE_k}\bra{e^-_{ij}}
\eal
with
\bal
p_{k,(i,j)}=
\begin{cases}
1 & k=2, \\
0 & k=1,3.
\end{cases}
\eal
Note that $K_{k, (i,j)}$'s only affect the diagonal elements.
Remarkably, $\sigma_0:= K_0 (\tilde\rho\otimes \tau) K_0^\dagger $ satisfies
\bal
(\Tr_S[\sigma_0])_{ij}=\tau_{ij}
\eal
for $(i,j)=(1,2), (2,3), (1,3), (1,1), (3,3)$, and 
\bal
(\Tr_S[\sigma_0])_{22}\leq \tau_{22}.
\eal
The terms with $K_{k, (i,j)}$ compensate the decrement in $(2,2)$ elements and recover the original state $\tau$ in the catalyst.
In addition, 
\bal
\rho_{13}'=\frac{152}{225}\delta^2>0
\eal
is also satisfied.

\end{proof}

One may wonder how to come up with this construction of states.
Here, we briefly describe the computation behind this choice.

We fix $\tilde\rho$ and $K_0$ as the aforementioned ones, and tune $\tau$.
We put restrictions $\tau_{11}=\tau_{33}$ and $\tau_{12}=\tau_{23}$ in our search.
Then, direct calculation of $K_0( \tilde\rho\otimes \tau) K_0^\dagger$ yields
\bal
(\Tr_S[\sigma_0])_{11}&=\frac12 \tau_{11}+\frac{1}{16}\tau_{22}+\frac{\delta}{2}\tau_{12}, \\
(\Tr_S[\sigma_0])_{22}&=\frac14 \tau_{11}+\frac58 \tau_{22}+\delta \tau_{12}, \\
(\Tr_S[\sigma_0])_{12}&=\frac{7}{16}\tau_{12}+\frac{3\delta}{4}\tau_{11}+\frac{\delta}{2}\tau_{22}+\frac{\delta}{4}\tau_{13}, \\
(\Tr_S[\sigma_0])_{13}&=\frac38 \tau_{13}+\frac{\delta}{2}\tau_{12}.
\label{tau0-computation}
\eal
In order to recover the original state of $C$, 
\bal
(\Tr_S[\sigma_0])_{11}&\leq \tau_{11}, \label{cond11} 
\eal
\bal
(\Tr_S[\sigma_0])_{22}&\leq \tau_{22}, \label{cond22} 
\eal
\bal
(\Tr_S[\sigma_0])_{12}&=\tau_{12}, \label{cond12} 
\eal
\bal
(\Tr_S[\sigma_0])_{13}&=\tau_{13} \label{cond13}
\eal
are necessary and sufficient.
The first two conditions are on diagonal elements, whose decrement can be compensated by $K_{k, (i,j)}$'s with proper $p_{k, (i,j)}$'s.
The last two conditions are on off-diagonal elements.

Substituting Eqs.~\eqref{cond12} and \eqref{cond13} into \eqref{tau0-computation}, and using $\tau_{22}=1-2\tau_{11}$, we have
\bal
\tau_{13}&=\frac{4\delta}{5}\tau_{12}, \\
\left( \frac{9}{16}-\frac{\delta^2}{5}\right) \tau_{12}&=\frac{\delta}{2}-\frac{\delta}{4}\tau_{11}. \label{line}
\eal
We should ensure the conditions \eqref{cond11} and \eqref{cond12}, which read
\bal
\frac54 \tau_{11}&\geq \frac18 +\delta \tau_{12}, \\
3&\geq 8\tau_{11}+8\delta \tau_{12}.
\eal
They imply that the coordinate $(\tau_{11}, \tau_{12})$ should settle in the triangle with three corners $(1/10, 0)$, $(3/8, 0)$, and $(2/9, 11/(18\delta))$.
Importantly, the area of this triangle increases as $\delta$ decreases, while the value of $\tau_{12}$ at $\tau_{11}=2/9$ on the line given by \eqref{line} decreases as $\delta $ decreases.
Hence, the line of \eqref{line} has some overlap with this triangle with small $\delta$. (In fact, $\delta\leq 1/4$ suffices to have the overlap.)
By adopting a particular choice $\frac54 \tau_{11}=\frac18 +\delta \tau_{12}$, which means the equality condition of \eqref{cond11}, we can solve $\tau_{12}$ and obtain \eqref{sigma}.
We remark that $\rho_{13}'$ is computed as
\bal
(\Tr_C[\sigma_0])_{13}=\frac14 \left(\frac32\tau_{13}+2\delta\tau_{12}\right),
\eal
which always takes a non-zero value in our setup.

\end{document}